\documentclass[pdflatex,sn-mathphys-num,icol]{sn-jnl}
\usepackage{fullpage}

\usepackage[T1]{fontenc}
\usepackage{amsmath,amssymb,amsthm}
\usepackage{mathtools}
\usepackage{graphicx}
\usepackage{booktabs}
\usepackage{xcolor}
\usepackage{comment}

\graphicspath{{BF-plots/}}
\theoremstyle{thmstyleone}
\newtheorem{theorem}{Theorem}[section]
\newtheorem{lemma}[theorem]{Lemma}
\newtheorem{corollary}[theorem]{Corollary}
\newtheorem{proposition}[theorem]{Proposition}

\theoremstyle{thmstyletwo}

\theoremstyle{thmstylethree}
\newtheorem{remark}[theorem]{Remark}

\begin{document}

\title[Passive two-plateau relaxation from Tricomi confluent hypergeometric kernels]{Passive two-plateau relaxation from Tricomi confluent hypergeometric kernels
\footnote{Published in: {\bf Nonlinear Dynamics, 114, 901 (2026). DOI: \href{https://doi.org/10.1007/s11071-026-12798-w}{10.1007/s11071-026-12798-w}.}}
}

\author[1,2]{\fnm{Marc} \sur{Tudela-Pi}}

\author[3]{\fnm{Ivano} \sur{Colombaro}}

\affil[1]{\orgdiv{Instituto de Microelectronica de Barcelona (IMB-CNM)}, \orgname{Consejo Superior de Investigaciones Cientificas (CSIC)}, \orgaddress{\city{Bellaterra}, \country{Spain}}}

\affil[2]{\orgname{CIBER-BBN, Instituto de Salud Carlos III}, \orgaddress{\city{Madrid}, \country{Spain}}}

\affil[3]{\orgdiv{Faculty of Engineering}, \orgname{Free University of Bozen-Bolzano}, \orgaddress{\street{via Bruno Buozzi 1}, \city{Bolzano}, \postcode{39100}, \country{Italy}}}

\abstract{%
Anomalous relaxation with memory spectra arises in disordered solids, soft matter, biological tissues and electrochemical interfaces. Fractional-order models capture broad power-law behaviour efficiently, but they can obscure spectral structure and are not always convenient for passive realisation or finite-dimensional simulation. We introduce a non-fractional passive framework based on the Tricomi confluent hypergeometric function, combined with a bounded Möbius normalisation that enforces prescribed low-frequency and high-frequency plateaux while preserving a broad dispersive transition. The resulting family contains the Debye and Cole–Cole responses as exact subcases, while extending them to asymmetric two-plateau dispersive laws with independently tunable low- and high-frequency exponents. For an admissible parameter range, we prove that the bounded block admits a Stieltjes representation with nonnegative spectral density, implying complete monotonicity, passivity, causality and compatibility with standard circuit, and state-space descriptions.
Building on this structure, we derive a passive Gauss–Stieltjes discretisation leading to Foster-type rational approximations and first-order state-space realisations with positive poles and residues. Numerical experiments show convergence of these finite-dimensional approximations across moderate-memory and long-tail regimes, enabling passive reduced-order representations of broad-memory responses. The framework is then validated on broadband dielectric data and battery electrochemical impedance spectra. In tissues, multi-block Tricomi mixtures improve complex-domain fitting accuracy relative to classical Cole–Cole baselines while preserving interpretable modal structure. In batteries, the same construction captures ageing-dependent spectral evolution and reveals a redistribution of dissipative dynamics toward slower characteristic time scales. Overall, the framework provides a constructive alternative to fractional-order models when passive realisability, spectral structure, and finite-dimensional implementation are required.}


\maketitle

\section{Introduction}

Linear response measurements in complex media frequently reveal non-exponential relaxation, in settings ranging from dielectric spectroscopy and polymer systems to bioimpedance and electrochemical impedance spectroscopy. Rather than a single characteristic decay time, experimental data point to broad and continuous distributions of relaxation times, leading to stretched transients in the time domain and frequency responses with flattened features and extended high-frequency tails. The classical single-time constant model, based on a first-order relaxation process, predicts an exponential decay that is adequate for dilute or weakly interacting systems, but it cannot account for the dispersive behaviour observed in disordered, heterogeneous, or strongly correlated materials. In such cases, relaxation is more naturally interpreted as the superposition of multiple microscopic processes with distinct characteristic scales, or equivalently as the action of a non-trivial memory kernel within a linear time-invariant framework. This phenomenology, commonly referred to as anomalous relaxation, motivates the search for modelling approaches that extend beyond the single-time constant paradigm while preserving linearity and passivity as required by physical realisability.

Fractional-order models have become a standard framework to describe anomalous relaxation beyond single-time-constant behaviour~\cite{Mainardi2010book}. By replacing classical first-order laws with derivatives of non-integer order, formulations such as fractional Cole--Cole~\cite{cole1941} or Havriliak--Negami~\cite{havriliak1967} elements reproduce stretched decays and power-law tails with a small number of parameters, and they are supported by a well-developed theory of completely monotone kernels and passive operators under suitable conditions~\cite{capelasoliveira2011MLrelaxation}. From a phenomenological standpoint, this approach provides an efficient and flexible description of broad distributions of relaxation times. At the same time, the fractional operators involved are intrinsically nonlocal in time, so that the response at a given instant is expressed through memory integrals over the past evolution of the system~\cite{tarasov2018NoNonlocality}. This nonlocality can make long-time time-domain simulations and integration into standard circuit or state-space solvers more demanding, and the interplay between fractional parameters and effective relaxation modes can lead to significant parameter degeneracy in practical identification.

An alternative viewpoint formulates anomalous relaxation directly through memory kernels or, equivalently, through distributions of relaxation times. Such kernel-based descriptions offer a natural connection with linear system theory, since passivity and causality can be enforced at the level of the kernel and the response can be interpreted as the superposition of elementary relaxation processes. For impedance applications, one would like these kernels to generate broad distributions of relaxation times and power-law-like tails while preserving a clear two-plateau behaviour with well-defined low-frequency and high-frequency limits that can be directly linked to measurable resistive contributions. Several families of functions have been explored to construct such kernels, including Mittag--Leffler~\cite{mainardi2015completemonotonicity} and Prabhakar~\cite{giusti2018prabhakar, giusti2020practicalguide} forms as well as confluent hypergeometric functions. Within this broader landscape, the Tricomi confluent hypergeometric function $U(a,b,z)$~\cite{tricomi1947funzioni} has been noted to produce non-exponential decays and power-law-type responses, suggesting it as a promising building block for dispersive models.
However, existing uses of these functions have mostly focused on analytical characterizations of relaxation and on classical properties of confluent hypergeometric functions~\cite{Slater1960confluent, Buchholz1969confluenthypergeometric, Erdelyi1953HTF-V1, Erdelyi1953HTF-V2}, or on their role as analytic solutions to differential equations in physics~\cite{mathews2022KummerGuide}, rather than on defining compact passive elements with an explicit memory kernel, an explicit spectral structure, and a representation compatible with standard circuit and state-space formulations. These features are important because they allow the resulting models to remain physically consistent and compatible with standard tools for circuit analysis and state-space modelling, while providing an explicit description of the underlying memory effects in a compact and interpretable form~\cite{Mainardi2010book, Widder1946Laplace}.

In this work we develop a non-fractional modelling framework in which anomalous relaxation is described by a memory kernel built from the Tricomi confluent hypergeometric function $U$ and normalised to match prescribed low-frequency and high-frequency limits. The resulting response laws generate broad relaxation spectra with one or two power-law regimes, contain Debye and Cole--Cole as exact subfamilies, and recover constant-phase-like asymptotics while remaining passive and completely monotone in suitable parameter ranges. Exploiting the special-function structure of $U$, we derive passive spectral representations, constructive rational approximations, and impedance/state-space realisations, so that these Tricomi-based kernels can be used as practical and interpretable building blocks for linear systems with long memory. The framework is then validated on broadband dielectric data from biological tissues and on electrochemical impedance spectra from lithium-ion batteries, illustrating its applicability across distinct passive dispersive systems.

These considerations motivate the present framework, which is designed to address three practical limitations often encountered in classical fractional-order models.
First, the intrinsic nonlocality of fractional operators leads to memory integrals over the full past history, making long-time simulations
computationally demanding. Second, fractional transfer functions do not generally admit a direct constructive passive realization in terms of finite-dimensional circuits or state-space systems. Third, their implementation in time-domain solvers typically requires either approximation or the introduction of auxiliary dynamics.

The Tricomi-based construction provides an alternative route that addresses these aspects simultaneously. In fact, the model is expressed through an explicit kernel admitting a Stieltjes representation, which ensures passivity and allows a direct constructive reduction to finite-dimensional positive-real systems. At the same time, the resulting formulation avoids explicit long-memory convolution in time-domain implementations once a reduced representation is obtained, while retaining the ability to describe broad dispersive relaxation behaviours.

\section{A framework for anomalous relaxation}
\label{sec:tricomi-framework}

\subsection{Anomalous relaxation with two plateaus}
\label{subsec:problem-setting}

We consider a causal linear time-invariant (LTI) block characterised by a kernel $k(t)$, defined for $t\ge 0$, and acting through convolution. For an input $x(t)$ and an output $y(t)$, one has
\begin{equation}
y(t)=\int_0^t k(t-\xi)\,x(\xi)\,\mathrm{d}\xi,
\end{equation}
so that, in the Laplace domain,
\begin{equation}
X(s)=\mathcal{L}\{x(t)\}(s), \qquad Y(s)=H(s)\,X(s),
\end{equation}
with
\begin{equation}
H(s)=\mathcal{L}\{k(t)\}(s)
=\int_0^\infty e^{-st}\,k(t)\,\mathrm{d}t,
\label{eq:memory-kernel-short}
\end{equation}
where $H(s)$ denotes the transfer law of the block and $k(t)$ its impulse response~\cite{Oppenheim1997signalssystems,Widder1946Laplace}.
More generally, the impulse response of the block may include a direct term, so that one may write
\[
h(t)=D\,\delta(t)+k(t),
\]
with corresponding transfer law
\[
H(s)=D+\mathcal{L}\{k(t)\}(s).
\]
In the present section, the convolution form is used to describe the strictly proper memory part, whereas the bounded two-plateau model introduced later may also include a direct high-frequency contribution.

In many complex media, the relevant transfer laws exhibit two distinct asymptotic regimes, namely a low-frequency plateau and a high-frequency plateau. These are defined by
\begin{equation}\label{eq:plateaus}
H_0 := \lim_{s\to 0} H(s),
\qquad
H_\infty := \lim_{|s|\to\infty} H(s).
\end{equation}
The limit $H_0$ in~\eqref{eq:plateaus} represents the quasi-static response of the system~\cite{Astrom2008feedbacksystem}, whereas $H_\infty$ characterises its fast-response behaviour~\cite{Khalil2002nonlinearsystems}. Under standard regularity and stability assumptions, these frequency-domain limits are linked to the long-time and short-time structure of the associated time-domain dynamics through Laplace-transform asymptotics~\cite{Ogata2010moderncontrol}, including the initial and final value theorems and related Tauberian arguments~\cite{Mainardi2010book,Widder1946Laplace}. Thus, the two plateaux of $H(s)$ may be viewed as the frequency-domain counterpart of the two asymptotic regimes typically observed in anomalous relaxation.

In anomalous relaxation, the kernel $k(t)$ is generally non-exponential and reflects the contribution of multiple effective relaxation scales rather than a single characteristic time. Equivalently, the corresponding transfer law $H(s)$ cannot be reduced to a single first-order rational form, but instead encodes broad memory effects through the structure of the kernel in~\eqref{eq:memory-kernel-short}. In this setting, properties such as nonnegativity and complete monotonicity of $k(t)$ play a central role in ensuring physical admissibility, causality and compatibility with passive realisations.

Our goal is therefore to construct kernels whose Laplace transforms generate broad dispersive transitions between the two plateaux in~\eqref{eq:plateaus}, while preserving analyticity in the right half-plane and a structure compatible with circuit and state-space descriptions. In the remainder of the paper, this transfer-law viewpoint will be specialised to impedance modelling, where the same bounded dispersive block will be anchored between prescribed low-frequency and high-frequency resistive plateaux.

\subsection{The Tricomi confluent hypergeometric function}
\label{subsec:tricomi}

Our construction is based on the Tricomi confluent hypergeometric function, usually denoted by $U(a,b,z)$, introduced by F.~Tricomi in 1947~\cite{tricomi1947funzioni}. The Tricomi confluent hypergeometric function is one of the two classical linearly independent solutions of the confluent hypergeometric differential equation, also known in the literature as Kummer's equation~\cite{Abramowitz1965handbook}, written
\begin{equation}
  z\,\frac{\mathrm{d}^2 w}{\mathrm{d} z^2}
  + (b - z)\,\frac{\mathrm{d} w}{\mathrm{d} z}
  - a\,w(z) = 0,
  \label{eq:tricomi_ode}
\end{equation}
with real parameters $a$ and $b$.
For completeness, we recall that equation~\eqref{eq:tricomi_ode} admits two linearly independent solutions, commonly denoted by the Kummer function $M(a,b,z)$~\cite[\S~47]{Oldham2009atlasfunctions} and the aforementioned Tricomi function $U(a,b,z)$~\cite[\S~48]{Oldham2009atlasfunctions}. The function $M(a,b,z)$ is entire in $z$ and regular at the origin, satisfying $M(a,b,0)=1$, whereas $U(a,b,z)$ exhibits a singular behaviour as $z\to 0$ for $b>1$. The two branches are connected through the classical Kummer relation~\cite{Erdelyi1953HTF-V1}
\begin{align}\label{eq:kummer-relation}
U(a,b,z) = \frac{\Gamma(1-b)}{\Gamma(a-b+1)}\, M(a,b,z) 
+ \frac{\Gamma(b-1)}{\Gamma(a)}\, z^{1-b} M(a-b+1,2-b,z),
\end{align}
which identifies $U(a,b,z)$ as the solution decaying along the positive real axis.

This qualitative behaviour makes $U(a,b,z)$ a natural candidate for modelling relaxation phenomena governed by a positive argument. Moreover, for parameters $a>0$ and $b>1$, the Tricomi function admits the Laplace-type integral representation~\cite{tricomi1947funzioni}
\begin{equation}
  U(a,b,z)
  = \frac{1}{\Gamma(a)}
    \int_{0}^{\infty}
      e^{-z t}\, t^{a-1} (1+t)^{\,b-a-1}\,\mathrm{d}t,
  \label{eq:tricomi_integral}
\end{equation}
which expresses $U(a,b,z)$ in the region $z>0$ as the Laplace transform of the nonnegative kernel
\begin{equation}
  g_{a,b}(t)
  := \frac{1}{\Gamma(a)}\, t^{a-1} (1+t)^{\,b-a-1},
  \qquad t>0.
  \label{eq:gab_def}
\end{equation}
Once the argument is identified with a rescaled Laplace variable of the form $z = s\tau$, the representation~\eqref{eq:tricomi_integral} makes it possible to interpret $U(a,b,s\tau)$ directly in terms of a time-domain kernel constructed from $g_{a,b}(t)$.

We will exploit this structure in two ways. First, the Laplace form~\eqref{eq:tricomi_integral} provides a direct route to define relaxation laws and, in particular, impedance elements in the frequency domain by inserting $U(a,b,s\tau)$ into suitable circuit topologies~\cite{VanValkenburg1964networkanalysis}. Second, the fact that $U(a,b,z)$ is generated by an explicit integral kernel will allow us, in later sections, to derive time-domain memory functions and to analyse conditions for passivity and complete monotonicity. For the purposes of the present subsection, it is enough to fix $U(a,b,z)$ as our basic special function and to record~\eqref{eq:tricomi_integral}--\eqref{eq:gab_def} as the main tools from which our Tricomi-based model will be constructed.

Once a rescaling of the form $z=s\tau$ is introduced, the Tricomi function in~\eqref{eq:tricomi_integral} appears as the Laplace transform of the nonnegative kernel~\eqref{eq:gab_def}, namely
\begin{equation}\label{eq:U_laplace_kernel}
    U(a,b,s\tau)
    = \int_{0}^{\infty} e^{-st}\, g_{a,b,\tau}(t)\,\mathrm{d}t,
    \quad s>0,
\end{equation}
with
\begin{equation}\label{eq:k_ab_tau}
    g_{a,b,\tau}(t)
    := \frac{\tau^{-1}}{\Gamma(a)}
       \left(\frac{t}{\tau}\right)^{a-1}
       \left(1+\frac{t}{\tau}\right)^{\,b-a-1},
\end{equation}
defined for $t>0$ and obtained from~\eqref{eq:gab_def} through the change of variables $t\mapsto t/\tau$, i.e.
\begin{equation}\label{eq:k_ab_tau-g}
    g_{a,b,\tau}(t)
    := \tau^{-1}\, g_{a,b}(t/\tau) .
\end{equation}
In particular, $g_{a,b,\tau}(t)$ is nonnegative and locally integrable on $(0,\infty)$, as shown in the next subsection.

\subsection{Structural properties of Tricomi-based relaxation kernels}
\label{subsec:tricomi_properties}

The integral representation~\eqref{eq:tricomi_integral}--\eqref{eq:gab_def} links the Tricomi confluent hypergeometric function $U(a,b,z)$ to memory kernels through a Laplace transform~\cite{tricomi1947funzioni, Slater1960confluent, Olver2010NIST}. To use $U(a,b,s\tau)$ as the core of a relaxation model, we briefly summarize a few structural properties that will be used later, namely analyticity in the right half-plane, the existence of a nonnegative time-domain kernel, complete monotonicity and BIBO stability.
These properties are crucial because, in impedance modelling, analyticity, causality, complete monotonicity and BIBO stability are consistent with physically meaningful behaviour. In particular, BIBO stability ensures that bounded inputs produce bounded outputs, preventing unphysical divergences. These properties are also compatible with passive realisations and with nonnegative memory-kernel representations. In the bounded Tricomi-based model developed below, passivity is established rigorously through the Stieltjes representation and the positive-real character of the corresponding transfer law.

\begin{lemma}[Analyticity]
\label{lem:analyticity}
Let $a>0$, $b>1$ and $\tau>0$. Then, the mapping
\begin{equation}
  s \;\mapsto\; U(a,b,s\tau)
\end{equation}
is analytic in the open half-plane $\{ s \in \mathbb{C} : \Re \{s\} > 0 \}$. Consequently, any transfer function obtained by composing $U(a,b,s\tau)$ with a rational function $\Psi$ that is analytic on a neighbourhood of its range is also analytic in $\Re \{s\}>0$.
\end{lemma}

\begin{proof}
For $a>0$ and $b>1$, $U(a,b,z)$ admits the Laplace-type representation in the region $\Re \{z\}>0$
\begin{equation}
  U(a,b,z)
  = \frac{1}{\Gamma(a)}
    \int_{0}^{\infty}
      e^{-z t}\, t^{a-1} (1+t)^{\,b-a-1}\,\mathrm{d}t,
\end{equation}
see, e.g.,~\cite{Slater1960confluent, Olver2010NIST}. The kernel is locally integrable and has at most algebraic growth, so the integral defines the Laplace transform of an exponentially bounded function. Standard results on the Laplace transform~\cite{Widder1946Laplace} then yield analyticity of $z \mapsto U(a,b,z)$ in $\Re \{z\}>0$, and hence of $s \mapsto U(a,b,s\tau)$ in $\Re \{s\}>0$. The statement for $\Psi(U(a,b,s\tau))$ follows by composition of analytic maps.
\end{proof}

\begin{lemma}[Complete monotonicity and nonnegative kernel]
\label{lem:cm}
Let $a>0$ and $b>1$. Define, for $t>0$,
\begin{equation}
  g_{a,b}(t)
  := \frac{1}{\Gamma(a)}\, t^{a-1}(1+t)^{\,b-a-1}.
\end{equation}
Then $g_{a,b}(t) \ge 0$ for all $t>0$ and, for every real $z>0$,
\begin{equation}
  U(a,b,z) = \int_0^\infty e^{-z t}\, g_{a,b}(t)\,\mathrm{d}t.
\end{equation}
In particular, $z \mapsto U(a,b,z)$ is completely monotone on $(0,\infty)$, and for any $\tau>0$ the function $s \mapsto U(a,b,s\tau)$, with $s>0$, is the Laplace transform of a nonnegative time-domain kernel.
\end{lemma}

\begin{proof}
The representation above shows that $U(a,b,z)$ is the Laplace transform of the nonnegative kernel $g_{a,b}(t)$. By Bernstein's theorem on completely monotone functions~\cite{Schilling2010bernstein}, this implies that $z \mapsto U(a,b,z)$ is completely monotone on $(0,\infty)$. The rescaling $z=s\tau$ with $\tau>0$ corresponds to a change of variables in the Laplace transform and yields a nonnegative time-domain kernel for $s \mapsto U(a,b,s\tau)$.
\end{proof}

We recall that a generic causal linear time-invariant system in the form~\eqref{eq:memory-kernel-short} with impulse response $k \in L^1_{\mathrm{loc}}([0,\infty))$ is said to be BIBO (bounded-input, bounded-output) stable if every bounded input produces a bounded output in the usual sense. A sufficient condition for BIBO stability is that the impulse response be integrable, $k \in L^1(0,\infty)$~\cite{Widder1946Laplace, Oppenheim1997signalssystems}.

\begin{corollary}[BIBO stability]
\label{cor:bibo}
Let $H(s)$ be a Tricomi-based relaxation law of the form
\begin{equation}
  H(s) = \int_0^\infty e^{-s t}\,g_{a,b,\tau}(t)\,\mathrm{d}t,
  \qquad \Re \{s\}>0,
\end{equation}
where $g_{a,b,\tau}(t) \ge 0$ for all $t>0$. Suppose that the low-frequency limit
\begin{equation}
  H_0 := \lim_{s \to 0^+} H(s)
\end{equation}
exists and is finite. Then $g_{a,b,\tau} \in L^1(0,\infty)$ and
\begin{equation}
  \int_0^\infty g_{a,b,\tau}(t)\,\mathrm{d}t = H_0,
\end{equation}
so the corresponding causal LTI system is BIBO stable.
\end{corollary}

\begin{proof}
For real $s>0$, one has
\(
 H(s) = \int_0^\infty e^{-s t}\,g_{a,b,\tau}(t)\,\mathrm{d}t
\)
with $g_{a,b,\tau}(t) \ge 0$. As $s \to 0^+$, the factors $e^{-s t}$ increase pointwise to $1$, so the monotone convergence theorem yields
\(
 \lim_{s\to 0^+} H(s) = \int_0^\infty g_{a,b,\tau}(t)\,\mathrm{d}t = H_0 < \infty.
\)
Thus $g_{a,b,\tau} \in L^1(0,\infty)$, which implies BIBO stability for the associated convolution system.
\end{proof}

For the bare Tricomi kernel $U(a,b,s\tau)$ with $b>1$, the low-frequency limit is generally divergent, so the corresponding response need not be BIBO-stable. This motivates the bounded Möbius normalisation introduced later in the paper. 

\subsection{Two algebraic regimes of the normalised mapping}
\label{subsec:two-power-laws}

To study the algebraic behaviour induced by the Tricomi function, we consider the normalised mapping
\begin{equation}\label{eq:WU_def}
\begin{aligned}
  W_U(s) &:= F_{a,b}(s\tau)=\Phi\!\big(U(a,b,s\tau)\big),
  \\
  \Phi(z) &:= \frac{z}{1+z},
\end{aligned}
\end{equation}
with parameters $a>0$, $b>1$ and a time scale $\tau>0$. The M\"obius map $\Phi$ interpolates between a unit low-frequency plateau, namely $\Phi(z)\to 1$ as $z\to\infty$, and a zero high-frequency plateau, $\Phi(z)\to 0$ as $z\to 0$.

Assuming $a>0$, $b>1$ on the positive real axis, the asymptotic behaviour of the Tricomi confluent hypergeometric function $U(a,b,z)$ satisfies the classical asymptotics~\cite{Slater1960confluent,Olver2010NIST}
\begin{gather}
  U(a,b,z) \overset{{z\to 0^+}}{\sim} \frac{\Gamma(b-1)}{\Gamma(a)} z^{\,1-b} \, , \label{eq:U_asympt0}\\
  U(a,b,z) \overset{{z\to +\infty}}{\sim} z^{-a} \,, \label{eq:U_asympt+infty}
\end{gather}
where the prefactor in~\eqref{eq:U_asympt0} is positive, i.e. $\Gamma(b-1)/\Gamma(a)>0$.

We can now identify the low-frequency regime by setting $z=s\tau$ and using~\eqref{eq:U_asympt0}, as $s\to 0^+$, so that $U(a,b,s\tau)\to +\infty$ and
\[
\Phi\!\big(U(a,b,s\tau)\big) \;=\; 1 - \frac{1}{U(a,b,s\tau)} + O\!\big(U^{-2}\big).
\]
Hence, we conclude that
\begin{equation}\label{eq:WU_low}
  W_U(s) \overset{{s\to 0^+}}{\sim} 1 - \frac{\Gamma(a)}{\Gamma(b-1)}(s\tau)^{b-1}.
\end{equation}

A similar argument applies to the asymptotic behaviour in the high-frequency regime.
Using $U(a,b,s\tau)\sim (s\tau)^{-a}$ as $|s|\to\infty$ with $\Re(s)>0$, and $\Phi\!\big(U(a,b,s\tau)\big)=U(a,b,s\tau)+O(U^2)$ as $U(a,b,s\tau)\to 0$, we obtain
\begin{equation}\label{eq:WU_high}
  W_U(s) \overset{{|s|\to+\infty}}{\sim} (s\tau)^{-a},
  \qquad
\ \Re(s)>0.
\end{equation}

The asymptotic expansions~\eqref{eq:WU_low} and~\eqref{eq:WU_high} show that the approach to the unit low-frequency plateau is governed by $b-1$, whereas the high-frequency decay is governed by $a$. These purely mathematical regimes will be used in Sect.~\ref{sec:phys} to build and analyze two-plateau linear elements without invoking fractional operators.

\section{A Tricomi-based impedance element with two resistive plateaux}
\label{sec:phys}

\subsection{Construction of a Tricomi-based impedance element}
\label{subsec:tricomi-impedance}

The normalised mapping $W_U(s)=F_{a,b}(s\tau)=\Phi\!\big(U(a,b,s\tau)\big)$ introduced in Sect.~\ref{subsec:two-power-laws} exhibits two algebraic regimes controlled independently by the exponents $a$ and $b-1$. The underlying Tricomi block $U(a,b,s\tau)$ is analytic in $\Re\{s\}>0$ and completely monotone for $s>0$, which makes $W_U$ a natural candidate for a bounded two-plateau response.

Before introducing the Tricomi-based construction, it is instructive to recall how the M\"obius map $\Phi(z)=z/(1+z)$ arises in the classical Debye circuit. In fact, let us consider a resistor $R_\infty$ in series with a parallel branch made of a resistor $R_\Delta$ and a capacitor $C$. The total impedance is
\begin{equation}
Z_{\mathrm{D}}(s)
= R_\infty + Z_\parallel(s),
\end{equation}
with
\begin{equation}
Z_\parallel(s)
= \left(\frac{1}{R_\Delta}+sC\right)^{-1},
\end{equation}
so that, setting $R_\Delta:=R_0-R_\infty$ and $\tau:=R_\Delta C$, we recover the classical Debye element~\cite{VanValkenburg1964networkanalysis}
\begin{equation}\label{eq:Zdebye}
Z_{\mathrm{D}}(s)
= R_\infty + \frac{R_0 - R_\infty}{1+s\tau}.
\end{equation}
This response admits the low-frequency plateau
\begin{equation}
\lim_{s\to 0} Z_{\mathrm{D}}(s)=R_0\,,
\end{equation}
and the high-frequency plateau
\begin{equation}
\lim_{|s|\to\infty} Z_{\mathrm{D}}(s)=R_\infty\,,
\end{equation}
and it is convenient to rewrite~\eqref{eq:Zdebye} using the dimensionless response in the Laplace domain
\begin{equation}
u(s):=\frac{1}{s\tau},
\end{equation}
for which
\begin{equation}
Z_{\mathrm{D}}(s)
= R_\infty + (R_0-R_\infty)\,\Phi\bigl(u(s)\bigr).
\label{eq:Z_Debye_F_form}
\end{equation}
Thus the rational mapping $\Phi$ is not arbitrary, but it is the M\"obius normalisation that naturally arises from the series--parallel reduction of a Debye branch expressed in terms of the scalar factor $u(s)=1/(s\tau)$~\cite{Oppenheim1997signalssystems,VanValkenburg1964networkanalysis}. This motivates the use of the same outer mapping $\Phi$ when replacing $u(s)$ by a more general response.

To embed this behaviour into a physically meaningful impedance element, we adopt the standard linear--fractional anchoring used in classical circuit reductions~\cite{Oppenheim1997signalssystems,VanValkenburg1964networkanalysis} to write
\begin{equation}\label{eq:ZU_def}
Z_U(s)
:=
R_\infty + (R_0 - R_\infty)\,W_U(s),
\end{equation}
with $R_0 > R_\infty > 0$. This preserves the prescribed low-frequency and high-frequency limits, respectively
\begin{gather}
\lim_{s\to 0} Z_U(s) = R_0, \label{eq:lim0ZU}\\
\lim_{|s|\to\infty} Z_U(s) = R_\infty, \label{eq:liminftyZU}
\end{gather}
independently of the parameters $a$, $b$, and $\tau$, and provides a bounded, causal, two-plateau response suitable for impedance modelling.

This construction generalizes the classical Debye element, recovered as the special case $a=1$ and $b=2$, for which $U(1,2,z)=z^{-1}$ and therefore
\begin{align}
W_U(s)&=\Phi\!\bigl((s\tau)^{-1}\bigr)=\frac{1}{1+s\tau}, \label{eq:comp-debye-W}\\
Z_U(s)&=R_\infty+\frac{R_0-R_\infty}{1+s\tau}. \label{eq:comp-debye-Z}
\end{align}
For any other choice of parameters in the admissible range $a>0$ and $b>1$, the internal Tricomi block $U(a,b,s\tau)$ no longer corresponds to a single exponential relaxation, and the resulting element displays a genuinely non-Debye transition between the two plateaux. This behaviour is consistent with other completely monotone non-exponential relaxation models characterised in circuit form~\cite{colombaro2026Besselbioimpedance}.

The analyticity and complete-monotonicity properties of the underlying Tricomi block make $W_U$ a natural candidate for a bounded two-plateau response.
The passivity of the bounded block are established in Sect.~\ref{sec:ss:stieltjes} for the admissible parameter range $a\in(0,1)$, $b\in(1,2)$, while the frequency-domain behaviour of $Z_U$ follows from the asymptotics of $W_U$ introduced in Sect.~\ref{subsec:two-power-laws} and is examined in detail in the next subsection.

\subsection{Frequency-domain characterisation}
\label{subsec:freq-characterisation}

The asymptotic regimes of the normalised mapping $W_U(s)$ previously established determine the qualitative structure of the frequency response of the Tricomi-based element $Z_U(s)$ defined in \eqref{eq:ZU_def}. Combining $Z_U(s)=R_\infty+(R_0-R_\infty)W_U(s)$ with the expansions~\eqref{eq:WU_low}--\eqref{eq:WU_high} yields, for real $s>0$,
\begin{align}
Z_U(s) & \overset{s \to 0^+}{\sim} R_0 - (R_0-R_\infty)\frac{\Gamma(a)}{\Gamma(b-1)}(s\tau)^{b-1},
 \label{eq:ZU_lowfreq}\\
Z_U(s) & \overset{|s|\to\infty}{\sim} R_\infty + (R_0-R_\infty)\,(s\tau)^{-a}. \label{eq:ZU_highfreq}
\end{align}
Thus, as already observed, the approach to the low-frequency plateau is governed by $b-1$, while the decay from the high-frequency plateau is governed independently by $a$.
This separation of roles allows the element to exhibit a wide range of algebraic transition behaviours between $R_0$ and $R_\infty$.

On the imaginary axis $s=\jmath\omega$, where $\omega$ denotes the angular frequency, the response traces a depressed arc in the Nyquist plane, associated with a single bounded dispersive relaxation block. As $\omega\to 0$, from~\eqref{eq:WU_low} we have
\begin{equation}
    W_U(\jmath\omega)\overset{\omega \to 0}{\sim} 1 - \frac{\Gamma(a)}{\Gamma(b-1)}\,(\jmath\omega\tau)^{\,b-1},
\end{equation}
so that the deviation $R_0-Z_U(\jmath\omega)$ scales as $\omega^{b-1}$ and has asymptotic phase $90(b-1)^\circ$.
Conversely, as $\omega\to\infty$, according to~\eqref{eq:WU_high} we get
\begin{equation}
W_U(\jmath\omega)\overset{\omega \to \infty}{\sim} (\jmath\omega\tau)^{-a},
\end{equation}
and the approach to the high-frequency plateau exhibits an algebraic decay with asymptotic phase $-90a$ degrees. These limits coincide with the behaviour of a fractional pseudo-capacitive element of order~$a$ at high frequency, while the exponent~$b-1$ controls the saturation towards the DC level.

\begin{proposition}[Non-Debye character of the Tricomi-based element]
\label{prop:non_debye}
Let $Z_U$ be the Tricomi-based impedance element defined in~\eqref{eq:ZU_def} with parameters $a>0$, $1<b\leq 2$, $R_0>R_\infty>0$ and $\tau>0$. Assume that $(a,b)\neq(1,2)$. Then $Z_U$ cannot be reduced to a single-time-constant Debye element of the form
\begin{equation}
  Z_{\mathrm{D}}(s)
  = R_\infty + \frac{R_0 - R_\infty}{1 + s\tau},
  \qquad \tau > 0,
  \label{eq:generic_Debye}
\end{equation}
and its time-domain relaxation is therefore non-exponential.
\end{proposition}

\begin{proof}
The Debye impedance~\eqref{eq:generic_Debye} satisfies the classical low-frequency and high-frequency expansions
\begin{align}
  Z_{\mathrm{D}}(s)
  &\overset{s\to 0}{=} R_0 - (R_0-R_\infty)\tau s + O(|s|^2)\, , \label{eq:Debye_lowfreq}\\
  Z_{\mathrm{D}}(s)
  &\overset{|s|\to\infty}{=} R_\infty + (R_0-R_\infty)\tau^{-1}s^{-1}
     + O(|s|^{-2}). \label{eq:Debye_highfreq}
\end{align}
In particular, the corrections around the plateaux scale as $|s|$ at low frequency and as $|s|^{-1}$ at high frequency.

On the other hand, for the Tricomi-based element $Z_U$, the asymptotics~\eqref{eq:ZU_lowfreq} and~\eqref{eq:ZU_highfreq} give corrections of order $|s|^{\beta}$ with $\beta=b-1$ as $s\to 0$, and of order $|s|^{-a}$ as $|s|\to\infty$. If $Z_U$ coincided with a Debye element~\eqref{eq:generic_Debye} for all $s$ with $\Re\{s\}>0$, these exponents would have to match, which forces $\beta=1$ and $a=1$. Since $\beta=b-1$, this implies $b=2$. Therefore, Debye behaviour can arise only in the special case $(a,b)=(1,2)$.

Finally, the Debye form~\eqref{eq:generic_Debye} corresponds to a single exponential relaxation in the time domain. Since, for $(a,b)\neq(1,2)$, the frequency response of $Z_U$ does not reduce to the Debye form, its associated time-domain response cannot be a single exponential either. As a consequence, the Tricomi-based element is genuinely non-Debye for all $(a,b)\neq(1,2)$.
\end{proof}

\begin{remark}
Throughout the previous discussion, the Tricomi kernel $U(a,b,z)$ is considered under the sufficient conditions $a>0$ and $b>1$, which ensure the existence of a nonnegative completely monotone time-domain kernel.
The additional restriction $1<b\le 2$ is introduced only from Proposition~\ref{prop:non_debye} onward, where the focus shifts from the general kernel to a bounded two-plateau relaxation element.
This range guarantees a physically meaningful low-frequency power-law regime, the existence of a finite DC plateau after normalisation, and it will later ensure a positive Stieltjes spectral density for the bounded mapping.
\end{remark}

To summarize, the Tricomi-based impedance $Z_U(s)$ provides a flexible two-plateau response in which the low- and high-frequency power laws are tunable and decoupled through the parameters $a$ and $b$. This contrasts with classical models in which a single exponent controls both tails and a detailed comparison with Debye and Cole--Cole elements is presented in the following Sect.~\ref{subsec:comparison}.

\subsection{Comparison with classical models}
\label{subsec:comparison}

The Tricomi-based element contains the Debye and Cole--Cole families as exact subcases. Outside those subfamilies it generates genuinely asymmetric two-plateau responses, with independently tunable low- and high-frequency exponents.

The Debye element is recovered for $a=1$ and $b=2$, as shown in~\eqref{eq:comp-debye-W}--\eqref{eq:comp-debye-Z}. In that case the response has the classical single-time-constant form, with a linear correction near the DC plateau and an $s^{-1}$ decay at high frequency.

A broader benchmark is provided by the Cole--Cole family~\cite{cole1941}, whose normalised impedance is
\begin{equation}
W_{\mathrm{CC}}(s)=\frac{1}{1+(s\tau_{\mathrm{CC}})^\alpha},
\qquad 0<\alpha\le 1.
\label{eq:WCC}
\end{equation}

\begin{proposition}[Exact embedding of the Cole--Cole family]
\label{prop:exact_cole_cole}
Let $a>0$ and set $b=a+1$. Then, for $\Re\{z\}>0$,
\begin{equation}
U(a,a+1,z)=z^{-a}.
\label{eq:U_aa1}
\end{equation}
Consequently,
\begin{equation}
W_U(s)=\frac{1}{1+(s\tau)^a},
\label{eq:WU_cole_cole}
\end{equation}
and the anchored impedance becomes
\begin{equation}
Z_U(s)=R_\infty+\frac{R_0-R_\infty}{1+(s\tau)^a}.
\label{eq:ZU_cole_cole}
\end{equation}
For $0<a\le 1$, this is precisely the classical Cole--Cole impedance with exponent $\alpha=a$ and characteristic time $\tau_{\mathrm{CC}}=\tau$. The Debye element is recovered as the endpoint $a=1$, $b=2$.
\end{proposition}

\begin{proof}
If $b=a+1$, the factor $(1+t)^{\,b-a-1}$ in~\eqref{eq:tricomi_integral} reduces to unity, so
\begin{equation}
\begin{aligned}
U(a,a+1,z)
&=
\frac{1}{\Gamma(a)}
\int_0^\infty e^{-zt} t^{a-1}\,\mathrm{d}t \\
&=
z^{-a},
\qquad \Re\{z\}>0.
\end{aligned}
\end{equation}
Substituting $z=s\tau$ into~\eqref{eq:WU_def} gives
\[
W_U(s)=\frac{(s\tau)^{-a}}{1+(s\tau)^{-a}}
      =\frac{1}{1+(s\tau)^a},
\]
and~\eqref{eq:ZU_cole_cole} follows from~\eqref{eq:ZU_def}.
\end{proof}

Proposition~\ref{prop:exact_cole_cole} shows that the relation $b=a+1$ is not merely an asymptotic matching condition: it embeds the Cole--Cole family exactly into the Tricomi construction. In particular, the symmetric case in which the low- and high-frequency exponents coincide is precisely the Cole--Cole subfamily.

Outside the manifold $b=a+1$, the asymptotics of $W_U$ remain
\begin{equation}
\begin{aligned}
W_U(s) &\overset{s\to 0^+}{\sim}
1-\frac{\Gamma(a)}{\Gamma(b-1)}(s\tau)^{\,b-1}, \\
W_U(s) &\overset{|s|\to\infty}{\sim}
(s\tau)^{-a},
\end{aligned}
\label{eq:WU_asym_comp}
\end{equation}
so the two algebraic regimes are governed independently by $b-1$ and $a$. Whenever $b\neq a+1$, the response cannot be reduced to a Cole--Cole element, which is constrained by a single exponent on both sides of the dispersive transition. The Tricomi-based element therefore provides a strict generalisation of Cole--Cole, retaining it as an exact one-parameter subfamily while allowing asymmetric power-law transitions between the two resistive plateaux.

The comparison above highlights the role of the Tricomi construction as a non-fractional relaxation block that contains the classical Debye and Cole--Cole models but also extends them to a wider class of bounded passive responses. In the next section we exploit the bounded mapping to derive a Stieltjes spectral representation and passive realisations suitable for circuit and state-space formulations.

\section{A passive bounded model based on the Tricomi function}
\label{sec:model}

After introducing the Tricomi-based impedance element and characterising its algebraic regimes and classical limits, we now focus on the passive bounded two-plateau model generated by the same Tricomi construction. More precisely, we restrict attention to the admissible parameter range in which the low-frequency divergence of the bare kernel is regularised by the M\"obius normalisation and the resulting element admits a Stieltjes spectral representation with nonnegative density. This provides the basis for establishing passivity and for deriving constructive rational and state-space realisations.

Ideal constant-phase elements (CPEs) correspond to power-law responses that, if extended down to arbitrarily low frequencies, imply unbounded long-time memory and ill-defined static limits. This motivates bounded regularizations that preserve the useful mid-band behaviour while enforcing bounded direct current (DC) and high-frequency asymptotes~\cite{holm2021CPE}.

Here we consider the same unity-feedback M\"obius normalisation introduced in~\eqref{eq:WU_def}, namely
\begin{equation}
F_{a,b}(z)=\Phi\!\bigl(U(a,b,z)\bigr)=\frac{U}{1+U}=1-\frac{1}{1+U},
\label{eq:F_def}
\end{equation}
where $\Phi(z)=z/(1+z)$ and $U=U(a,b,z)$. We use the variable $z$ here to emphasize that the construction is not restricted to the specific argument $z=s\tau$, although that is the case of interest for transfer laws and impedance elements.

For $b\in(1,2)$, the Tricomi function exhibits the low-frequency scaling along the imaginary axis
\begin{equation}
U(a,b,\jmath\omega\tau) \overset{\omega\to 0}{\sim} (\jmath\omega\tau)^{\,1-b},
\qquad b\in(1,2),
\label{eq:U_LF}
\end{equation}
so that $|U|\to\infty$ as $\omega\to 0$, and $U$ alone cannot define a bounded two-plateau response with a finite static limit. The normalised mapping $F_{a,b}$ removes this divergence by saturating the low-frequency behaviour to a finite plateau. Moreover, from~\eqref{eq:U_LF}, one has $F_{a,b}(\jmath\omega\tau)\to 1$ as $\omega\to 0$ for $b\in(1,2)$. At the opposite extreme, the large-argument behaviour of Tricomi $U$ yields $U(a,b,z)\sim z^{-a}$ as $|z|\to\infty$ on the principal branch~\cite{Abramowitz1965handbook,Olver2010NIST}, so $F_{a,b}(\jmath\omega\tau)\to 0$ as $\omega\to\infty$.

To use $F_{a,b}$ as a generic two-plateau building block, the asymptotes are anchored explicitly by the transfer function
\begin{equation}
H(\jmath\omega)=H_\infty+\bigl(H_0-H_\infty\bigr)\,F_{a,b}(\jmath\omega\tau),
\label{eq:H_two_plateau}
\end{equation}
so that $H(\jmath\omega)\to H_0$ as $\omega\to 0$ and $H(\jmath\omega)\to H_\infty$ as $\omega\to\infty$. For time-domain analysis we use the analytic continuation to the right half-plane and write
\[
H(s)=H_\infty+(H_0-H_\infty)F_{a,b}(s\tau),
\qquad \Re\{s\}>0.
\]

\subsection{Spectral Stieltjes representation and passivity}
\label{sec:ss:stieltjes}

To establish the passivity of the bounded Tricomi-based model introduced above, we show that the mapping $F_{a,b}(z)$ is a Stieltjes function for a suitable parameter range. In particular, passivity follows once a nonnegative spectral density is identified on the branch cut of $F_{a,b}$.

We work under the parameter values $a\in(0,1)$ and $b\in(1,2)$, using the principal branch for complex powers and the cut $(-\infty,0]$. For $z\in\mathbb{C}\setminus(-\infty,0]$ the spectral density on the cut is defined through the Stieltjes--Perron inversion
\begin{equation}
\rho_F(x):=-\frac{1}{\pi}\Im\!\left\{F_{a,b}\!\bigl(-x+\jmath 0^+\bigr)\right\},\; x>0.
\label{eq:rho_def}
\end{equation}

The following result ensures nonnegative spectral density on the cut and therefore establishes the Stieltjes structure of $F_{a,b}$.

\begin{proposition}[Nonnegative density and passivity]
\label{prop:rho_nonneg}
For $a\in(0,1)$ and $b\in(1,2)$, the density~\eqref{eq:rho_def} satisfies $\rho_F(x)\ge 0$ for all $x>0$.
Consequently, $F_{a,b}$ admits the Stieltjes representation
\begin{equation}
F_{a,b}(z)=\int_{0}^{\infty}\frac{\rho_F(x)}{x+z}\,\mathrm{d}x,
\qquad z\in\mathbb{C}\setminus(-\infty,0],
\label{eq:F_stieltjes}
\end{equation}
and is positive-real for $\Re\{z\}>0$
\begin{equation}
\Re\{F_{a,b}(z)\}=\int_0^\infty \rho_F(x)\,\Re\!\left\{\frac{1}{x+z}\right\}\mathrm{d}x \ge 0.
\label{eq:PR_F}
\end{equation}
\end{proposition}

Together,~\eqref{eq:F_stieltjes}--\eqref{eq:PR_F} yield a complete Stieltjes structure for $F_{a,b}$ and therefore guarantee passivity and stability of the model.
A full proof of Proposition~\ref{prop:rho_nonneg} is given in Appendix~\ref{app:rho_proof}.

\subsection{Parameter effects in the anchored two-plateau response}
\label{sec:ss:param_effects}

Having established passivity, the qualitative effect of the parameters on the anchored response~\eqref{eq:H_two_plateau} is properly summarized in Fig.~\ref{fig:F_properties}, where we notice that the approach to the plateaus is algebraic and separates the roles of $a$ and $b$.

When $|U|\gg 1$, so for low frequencies, $F_{a,b}=1-(1+U)^{-1}=1-U^{-1}+o(U^{-1})$, and using~\eqref{eq:U_LF} yields
\begin{equation}
1-F_{a,b}(\jmath\omega\tau)\overset{\omega\to 0}{\sim} \frac{1}{U(a,b,\jmath\omega\tau)}\overset{\omega\to 0}{\sim} (\jmath\omega\tau)^{\,b-1}.
\label{eq:LF_tail}
\end{equation}
Substituting into~\eqref{eq:H_two_plateau} gives $H_0-H(\jmath\omega)=(H_0-H_\infty)(1-F)$, so that
\begin{equation}
|H_0-H(\jmath\omega)| \propto \omega^{\,b-1},
\;
\angle\!\bigl(H_0-H(\jmath\omega)\bigr) \to 90(b-1)^\circ,
\label{eq:LF_slope_phase}
\end{equation}
i.e. the exponent $b-1$ governs the algebraic approach to the DC plateau.

Conversely, when $|U|\ll 1$, namely in the high-frequency regime, $F_{a,b}=U+O(U^2)$ and $U(a,b,z)\sim z^{-a}$ yields
\begin{equation}
F_{a,b}(\jmath\omega\tau) \overset{\omega\to \infty}{\sim} (\jmath\omega\tau)^{-a},
\label{eq:HF_tail}
\end{equation}
so that $H(\jmath\omega)-H_\infty=(H_0-H_\infty)F$ and therefore
\begin{equation}
|H(\jmath\omega)-H_\infty| \propto \omega^{-a},
\;
\angle\!\bigl(H(\jmath\omega)-H_\infty\bigr) \to -90a^\circ,
\label{eq:HF_slope_phase}
\end{equation}
i.e. $a$ controls the algebraic approach to the high-frequency plateau.
The anchoring parameters $H_0$ and $H_\infty$ set plateau levels and the step amplitude without altering the intrinsic shape induced by $a$ and $b$, while $\tau$ shifts the dispersion along the frequency axis.
In computations,~\eqref{eq:F_def} is evaluated as $F_{a,b}=1-(1+U)^{-1}$ to avoid loss of significance when $|U|\gg 1$.

The resulting low- and high-frequency slopes and asymptotic phase limits in~\eqref{eq:LF_slope_phase} and~\eqref{eq:HF_slope_phase} are fully consistent with the frequency-domain behaviour of the normalised mapping $W_U(\jmath\omega)$ discussed in Sect.~\ref{subsec:freq-characterisation}.
\begin{figure}[!htb]
    \centering
    \includegraphics[
        width=0.9\textwidth,
        height=0.78\textheight,
        keepaspectratio
    ]{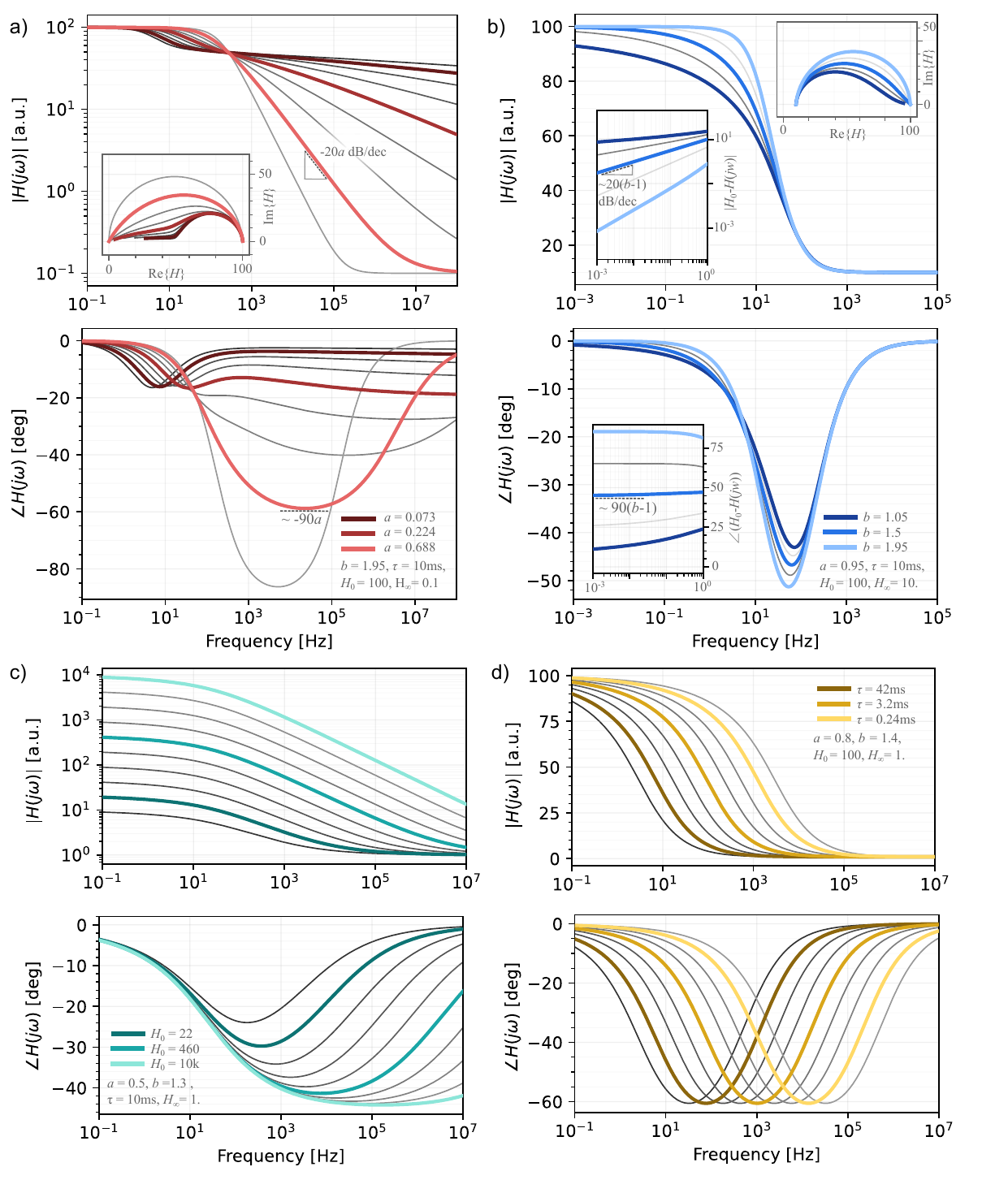}
    \caption[Parameter effects in the two-plateau response]{Parameter effects in the two-plateau response $H(\jmath\omega)=H_\infty+(H_0-H_\infty)F_{a,b}(\jmath\omega\tau)$. Panels illustrate how each parameter shapes the magnitude, phase, Nyquist profile and asymptotic slopes.
    (a) Varying $a$ (with $b$ fixed) controls the high-frequency algebraic tail:
$H(\jmath\omega)-H_\infty \propto \omega^{-a}$ with asymptotic phase $\to -90a^\circ$, slope $\approx -20a$ dB/dec.
    (b) Varying $b$ (with $a$ fixed) controls the low-frequency approach to the DC plateau: $H_0-H(\jmath\omega)\propto \omega^{\,b-1}$ with phase $\to 90(b-1)^\circ$, slope $\approx 20(b-1)$ dB/dec.
    (c) Varying $H_0$ rescales the step amplitude without changing the intrinsic dispersion shape.
    (d) Varying $\tau$ shifts the dispersion along the frequency axis (fixed $a,b,H_0,H_\infty$).}
    \label{fig:F_properties}
\end{figure}

\subsection{Exact time-domain response}
\label{sec:ss:time_exact}

The Stieltjes representation~\eqref{eq:F_stieltjes} provides a direct and explicit time-domain interpretation of the bounded Tricomi block.
Since $\rho_F(x)\ge 0$, the mapping
\begin{equation}\label{eq:Fstielties}
F_{a,b}(s\tau)=\int_0^\infty \frac{\rho_F(x)}{x+s\tau}\,\mathrm{d}x\,,
\qquad\Re\{s\}>0    \,,
\end{equation}
is the Laplace transform of a nonnegative completely monotone kernel. Writing $(x+s\tau)^{-1}=\tau^{-1}(s+x/\tau)^{-1}$ and inverting termwise gives the impulse response of the bounded block as an exact mixture of decaying exponentials, namely
\begin{equation}
\begin{aligned}
f(t)
:= \mathcal{L}^{-1}\!\{F_{a,b}(s\tau)\}(t) 
= \frac{1}{\tau}\int_0^\infty \rho_F(x)\,e^{-xt/\tau}\,\mathrm{d}x,
\qquad t>0.
\end{aligned}
\label{eq:f_kernel}
\end{equation}
Hence $f(t)$ is completely monotone and provides a direct time-domain signature of passivity and stability.

For reference, the underlying Tricomi kernel associated with $U(a,b,s\tau)$ is given by the integral representation~\eqref{eq:tricomi_integral}, which yields the explicit closed-form impulse response for $t>0$
\begin{equation}
\begin{aligned}
u_U(t)
= \mathcal{L}^{-1}\!\{U(a,b,s\tau)\}(t) = \frac{1}{\tau\Gamma(a)}
\left(\frac{t}{\tau}\right)^{a-1}
\left(1+\frac{t}{\tau}\right)^{b-a-1},
\end{aligned}
\label{eq:uU_kernel}
\end{equation}
coinciding exactly with the rescaled Tricomi kernel $g_{a,b,\tau}(t)$ introduced in~\eqref{eq:k_ab_tau}. Thus, the underlying Tricomi kernel is nonnegative, while the bounded block acquires complete monotonicity through its own Stieltjes representation with nonnegative spectral density.

The anchored two-plateau response in~\eqref{eq:H_two_plateau} follows from the linearity of the Laplace transform. Indeed, taking $H(s)=H_\infty+(H_0-H_\infty)F_{a,b}(s\tau)$, and inverting it termwise gives
\begin{equation}
h(t):=\mathcal{L}^{-1}\!\{H(s)\}(t)
=H_\infty\,\delta(t)+(H_0-H_\infty)\,f(t),
\label{eq:h_kernel}
\end{equation}
which is defined in the sense of distributions for $t\ge 0$. The term $H_\infty\,\delta(t)$ encodes the instantaneous high-frequency contribution, whereas $f(t)$ describes the dispersive relaxation generated by the bounded Tricomi block.

It is often useful to derive the associated step response. Since $\displaystyle\frac{1}{x+s\tau}\cdot \frac{1}{s}$ inverts $\displaystyle\frac{1-e^{-xt/\tau}}{x}$, from~\eqref{eq:Fstielties} we obtain
\begin{equation}
\begin{aligned}
g(t)
:= \mathcal{L}^{-1}\!\{F_{a,b}(s\tau)/s\}(t) = \int_0^\infty \rho_F(x)\,\frac{1-e^{-xt/\tau}}{x}\,\mathrm{d}x,
\quad t\geq 0,
\end{aligned}
\label{eq:g_kernel}
\end{equation}
where $g(t)$ is now the step response of the bounded block.

Differentiating under the integral sign gives $g'(t)=f(t)$, which is fully consistent with~\eqref{eq:f_kernel}. Finally, substituting $F_{a,b}(s\tau)/s$ into the anchored model yields
\begin{equation}
\mathcal{L}^{-1}\!\{H(s)/s\}(t)
=H_\infty\,\Theta(t)+(H_0-H_\infty)\,g(t),
\end{equation}
providing an explicit and passive time-domain realisation of both the impulse and step responses of the bounded two-plateau model defined for $t\ge 0$,  where $\Theta(t)$ denotes the Heaviside step function with $\Theta(t)=0$ for $t<0$ and $\Theta(t)=1$ for $t>0$.

\subsection{Rational approximation and first-order state-space realisation}
\label{sec:ss:ss_realisation}

Evaluating~\eqref{eq:F_stieltjes} at $z=0$ and using $F_{a,b}(0)=1$ yields
\begin{equation}
\int_0^\infty \frac{\rho_F(x)}{x}\,\mathrm{d}x = 1.
\end{equation}
This motivates the normalised positive measure
\begin{equation}
\mathrm{d}P(x):=\frac{\rho_F(x)}{x}\,\mathrm{d}x,
\label{eq:P_measure}
\end{equation}
on $(0,\infty)$.
Hence, using the identity
\[
\frac{x}{x+z}=1-\frac{z}{x+z},
\]
and the normalisation
\[
 \int_0^\infty dP(x)=1,
\]
we obtain the exact decomposition
\begin{equation}
F_{a,b}(z)=\int_0^\infty \frac{x}{x+z}\,\mathrm{d}P(x)
=1-z\,S(z),
\label{eq:ss:F_from_S}
\end{equation}
defining
\begin{equation}
S(z):=\int_0^\infty \frac{\mathrm{d}P(x)}{x+z},
\label{eq:Sdef}
\end{equation}
valid for $z\in\mathbb{C}\setminus(-\infty,0]$.
This formulation separates the unit DC plateau $F_{a,b}(0)=1$ from the non-trivial dispersive behaviour encoded by the transform $S(z)$.

The function $S(z)$ also admits a convenient integral representation on the log--rate axis. Introducing the variable $u=\ln x$ as seen in the previous section, the measure $\mathrm{d}P(x)$ becomes $p(u)\,\mathrm{d}u$ with $p(u)=\rho_F(e^u)$,
so that
\begin{equation}
S(z)=\int_{-\infty}^{\infty} \frac{p(u)}{e^{u}+z}\,\mathrm{d}u.
\label{eq:S_log_form}
\end{equation}
This representation avoids Jacobian amplification and is well suited to numerical quadrature on extended log-ranges.

Let $\{(x_m,w_m)\}_{m=1}^{N}$ denote the Gauss--Stieltjes nodes and weights associated with the positive measure $\mathrm{d}P(x)$, equivalently with its log-rate form $p(u)\,\mathrm{d}u$ under the change of variables $u=\ln x$ \cite{gautschi2004orthogonal}. Applying Gauss quadrature to~\eqref{eq:S_log_form} yields
\begin{equation}
S(z)\approx S_N(z):=\sum_{m=1}^{N} \frac{w_m}{x_m+z},
\label{eq:SN_def}
\end{equation}
being $x_m>0$ and $w_m>0$, and satisfying $\sum_{m=1}^{N} w_m=1$.
In practice, the spectral density $\rho_F(x)$ is evaluated numerically from the boundary values of the Tricomi function along the branch cut, and standard Gauss--Stieltjes or Gauss--Jacobi quadrature routines can  be employed to generate the corresponding nodes and weights.
Substituting now~\eqref{eq:SN_def} into~\eqref{eq:ss:F_from_S} gives the rational approximation
\begin{equation}
F_{a,b}(z)\approx F_N(z):=1-z\,S_N(z)
=\sum_{m=1}^{N} \frac{w_m x_m}{x_m+z},
\label{eq:FN_def}
\end{equation}
which is a Foster-type positive-real expansion, namely all poles $z=-x_m$ lie on the negative real axis and all residues $w_m x_m$ are nonnegative. The DC plateau is matched exactly since $F_N(0)=\sum_m w_m=1$.

A corresponding state-space realisation is obtained by writing $z=s\tau$ and defining
\[
\lambda_m:=\frac{x_m}{\tau},
\qquad
r_m:=\frac{w_m x_m}{\tau},
\]
the approximation~\eqref{eq:FN_def} becomes
\begin{equation}
F_N(s\tau)=\sum_{m=1}^{N}\frac{r_m}{s+\lambda_m},
\qquad \lambda_m>0,\; r_m>0.
\label{eq:FN_state_form}
\end{equation}
This admits the diagonal realisation
\begin{equation}
\begin{aligned}
\mathbf{A} &= -\mathrm{diag}(\lambda_1,\dots,\lambda_N), \\
\mathbf{B} &= \mathbf{1}, \\
\mathbf{C} &= \begin{bmatrix} r_1 & \cdots & r_N \end{bmatrix}, \\
D &= 0,
\end{aligned}
\label{eq:ss_realisation}
\end{equation}
so that $F_N(s\tau)=\mathbf{C}(s\mathbf{I}-\mathbf{A})^{-1}\mathbf{B}+D$.
Stability follows from $\lambda_m>0$, and passivity follows from the positive-real Foster form~\eqref{eq:FN_def}.

The same approximation also admits the standard Foster circuit interpretation shown in Fig.~\ref{fig:foster_rc}. When the anchored response is interpreted as an impedance, the positive-real expansion corresponds to a finite series connection of passive parallel $R_m\parallel C_m$ {branches, together with the high-frequency term} $Z_\infty$. The positivity of the Gauss--Stieltjes nodes and weights guarantees positive poles and residues, and therefore positive circuit parameters.

\begin{figure}[t]
    \centering
    \includegraphics[width=0.75\linewidth]{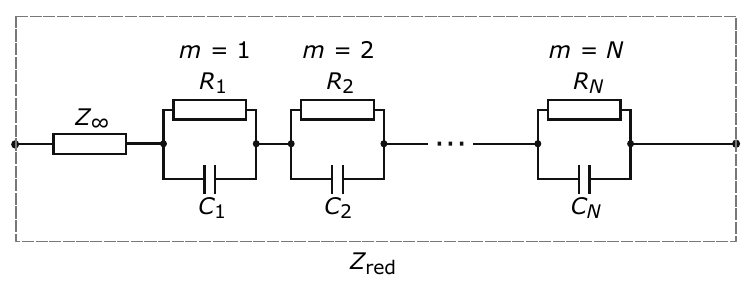}
    \caption{Foster-type passive realisation of the Gauss--Stieltjes reduced Tricomi block. When the anchored response is interpreted as an impedance, the reduced positive-real approximation can be represented as a series connection of passive $R_m\parallel C_m$ branches, together with the high-frequency term $Z_\infty$. Since the Gauss--Stieltjes construction yields positive poles and residues, all circuit parameters are positive and the finite-dimensional approximation preserves causality, stability and passivity by construction.}
    \label{fig:foster_rc}
\end{figure}

For the anchored approximation, inserting~\eqref{eq:FN_state_form} into the anchored model~\eqref{eq:H_two_plateau} yields
\begin{equation}
H(s)\approx H_{\infty}
+\bigl(H_{0}-H_{\infty}\bigr)F_N(s\tau),
\label{eq:H_approx}
\end{equation}
providing a passive $N$-mode realisation of the bounded two-plateau response.
Multiple dispersive processes can be accommodated by summing independent blocks with their own parameters $a_n$, $b_n$, $\tau_n$ and amplitudes $\Delta H_n$, without affecting stability or passivity.

\section{Methods and computational pipeline}
\label{sec:methods}

\subsection{Tissue dielectric data and preprocessing}
\label{sec:methods:tissues}

We validated the proposed bounded Tricomi block on broadband tissue dielectric measurements from the Gabriel dataset, using the tabulated experimental data reported in the measurement study of Gabriel and co-workers~\cite{gabriel1996ii}. For benchmark comparison, the classical Cole--Cole-type baseline was taken from the IFAC-CNR online resource~\cite{andreuccetti2012}, which implements the parametric Gabriel model rather than the original tabulated measurements.

For each tissue, the available data were given as relative permittivity $\varepsilon'(\omega)$ and effective conductivity $\sigma(\omega)$ over a frequency grid $\{f_k\}_{k=1}^{K}$, with $\omega_k = 2\pi f_k$. To keep the formulation consistent with dispersive dielectrics, we work with the complex relative permittivity
\begin{equation}
\hat{\varepsilon}_{\mathrm{exp}}(\omega_k)
=
\varepsilon'_{\mathrm{exp}}(\omega_k)
+\frac{\sigma_{\mathrm{exp}}(\omega_k)}{\jmath\omega_k\varepsilon_0},
\label{eq:methods:epshat_exp}
\end{equation}
where $k=1,\dots,K$ and $\varepsilon_0$ denotes the vacuum permittivity.

Prior to fitting, the data were filtered to remove non-finite entries and non-positive frequencies, and were then sorted in increasing order of frequency.

When an impedance-domain visualisation was desired, the fitted dielectric response was mapped to a homogeneous slab of thickness $L$ and electrode area $A$ through the complex admittance
\begin{equation}
\begin{aligned}
Y(\omega)
&=
\frac{A}{L}\Bigl(
\sigma_{\mathrm{model}}(\omega)
+ \jmath\omega\varepsilon_0\,\varepsilon'_{\mathrm{model}}(\omega)
\Bigr), \\
Z(\omega)
&= \frac{1}{Y(\omega)}.
\end{aligned}
\label{eq:methods:slab_impedance}
\end{equation}
This mapping was used only for visualisation and does not affect the fitting procedure, which was carried out in the dielectric domain.

\subsection{Multi-block Tricomi permittivity model}
\label{sec:methods:eps_model}

To model tissue dispersion with multiple relaxation processes, we superposed $N$ bounded Tricomi blocks. Using the normalised mapping $F_{a,b}(z)=U(a,b,z)/(1+U(a,b,z))$, the complex relative permittivity was written as
\begin{equation}
\hat{\varepsilon}_{\mathrm{model}}(\omega)
=
\varepsilon_{\infty}
+
\frac{\sigma_i}{\jmath\omega\varepsilon_0}
+
\sum_{n=1}^{N}
\Delta\varepsilon_n\,
F_{a_n,b_n}(\jmath\omega\tau_n),
\label{eq:methods:epshat_model}
\end{equation}
where $\varepsilon_{\infty}$ denotes the high-frequency plateau, $\sigma_i$ an ionic-conduction term, and each dispersive block is parametrised by $\Delta\varepsilon_n$, $a_n$, $b_n$, $\tau_n$ with $\Delta\varepsilon_n>0$ and $\tau_n>0$. In agreement with the passive Stieltjes regime established in Sect.~\ref{sec:ss:stieltjes}, the shape parameters were constrained to $a_n\in(0,1)$ and $b_n\in(1,2)$.

The model observables were obtained as
\begin{equation}
\begin{aligned}
\varepsilon'_{\mathrm{model}}(\omega)&=\Re\{\hat{\varepsilon}_{\mathrm{model}}(\omega)\},
\\
\sigma_{\mathrm{model}}(\omega)&=-\omega\varepsilon_0\,\Im\{\hat{\varepsilon}_{\mathrm{model}}(\omega)\}.
\label{eq:methods:eps_sigma_from_epshat}
\end{aligned}
\end{equation}
For the tissue study, candidate model orders $N=2,3,4$ were fitted and compared.

\subsection{Parameter estimation and weighting across decades}
\label{sec:methods:fit}

Parameters are estimated by minimising a weighted discrepancy in decades (log-ratio errors), which is well suited to the large dynamic range of dielectric spectra.
For each frequency sample $\omega_k$ we define
\begin{equation}
\begin{aligned}
e_{\varepsilon}(k)
&= \log_{10}\!\left(
\frac{\varepsilon'_{\mathrm{model}}(\omega_k)}
{\varepsilon'_{\mathrm{exp}}(\omega_k)}
\right), \\
e_{\sigma}(k)
&= \log_{10}\!\left(
\frac{\sigma_{\mathrm{model}}(\omega_k)}
{\sigma_{\mathrm{exp}}(\omega_k)}
\right),
\end{aligned}
\label{eq:methods:log_errors}
\end{equation}
after applying small positive floors to avoid logarithms of non-positive values.
These errors are stacked into a residual vector
\begin{equation}
\mathbf{r}(\boldsymbol{\theta})
=
\begin{bmatrix}
W_{\varepsilon}\,\mathbf{W}^{1/2}\, \mathbf{e}_{\varepsilon}\\[2pt]
W_{\sigma}\,\mathbf{W}^{1/2}\, \mathbf{e}_{\sigma}
\end{bmatrix},
\label{eq:methods:residual_vector}
\end{equation}
where $\boldsymbol{\theta}$ collects all model parameters, $(W_{\varepsilon},W_{\sigma})$ are user-defined relative weights, and $\mathbf{W}=\mathrm{diag}(w_1,\dots,w_K)$ applies per-frequency weighting.
To balance the contribution of each frequency decade, we set $w_k \propto \Delta\log_{10}(f_k)$, so that the resulting sum of squares approximates a Riemann integral over $\log_{10}(f)$, and for log-uniform grids this reduces to nearly uniform weights.

Optimisation is performed in two stages. First, a global search using differential evolution provides a robust initialisation on a log-subsampled frequency subset. Second, a local refinement uses robust nonlinear least squares with a soft-$\ell_1$ loss on the full band. To prevent label switching between blocks, the time constants $\tau_n$ are sorted after each parameter update, and the corresponding $\Delta\varepsilon_n$, $a_n$ and $b_n$ are permuted consistently. Parameter bounds are enforced on $\varepsilon_{\infty}$, $\sigma_i$, $\Delta\varepsilon_n$, $a_n$, $b_n$ and $\tau_n$, while amplitudes and time constants are optimised in base-10 logarithmic coordinates.

\subsection{Complex-domain validation, model selection and bootstrap mode analysis}
\label{sec:methods:validation_bootstrap}

Although parameter estimation for tissues was performed on logarithmic residuals of $\varepsilon'(\omega)$ and $\sigma(\omega)$, model comparison was additionally assessed in the complex domain. For each fitted model, we evaluated the pointwise relative complex error
\begin{equation}
e_{\mathrm{c}}(\omega_k)
=
\frac{
\left|
\hat{\varepsilon}_{\mathrm{model}}(\omega_k)
-
\hat{\varepsilon}_{\mathrm{exp}}(\omega_k)
\right|
}{
\left|
\hat{\varepsilon}_{\mathrm{exp}}(\omega_k)
\right|
},
\quad k=1,\dots,K,
\label{eq:methods:complex_error}
\end{equation}
and summarised it through the root-mean-square error
\begin{equation}
\mathrm{RMSE}_{\mathrm{c}}
=
\left(
\frac{1}{K}\sum_{k=1}^{K} e_{\mathrm{c}}(\omega_k)^2
\right)^{1/2}.
\label{eq:methods:complex_rmse}
\end{equation}
This quantity was used only for cross-model comparison on the experimental Gabriel frequency grid.

For model-order comparison within the Tricomi family, comparative information criteria were computed from the same weighted logarithmic residual vector used in the fit, namely Eq.~\eqref{eq:methods:residual_vector}. In all tissue fits, the relative residual weights were fixed to $W_{\varepsilon}=1.5$ and $W_{\sigma}=1.0$, and the per-frequency weights were chosen proportional to the local spacing in $\log_{10}f$. The resulting weighted residual sum of squares was used to compute comparative AIC- and BIC-type criteria across candidate model orders.

For modal post-processing in tissues, each fitted Tricomi block was represented by an in-band frequency-localised profile. For the $n$-th block, we defined the raw profile
\begin{equation}
\phi_n(f)
=
\Delta\varepsilon_n
\max\!\left\{
-\Im\!\left[
F_{a_n,b_n}(\jmath 2\pi f\tau_n)
\right], 0
\right\},
\label{eq:methods:tissue_mode_raw}
\end{equation}
and the corresponding normalised mode-density profile
\begin{equation}
\psi_n(f)
=
\frac{\phi_n(f)}
{\displaystyle
\int_{\log_{10}f_{\min}}^{\log_{10}f_{\max}}
\phi_n(\xi)\,d\log_{10}\xi }.
\label{eq:methods:tissue_mode_profile}
\end{equation}
By construction, $\psi_n$ has unit area over the measured frequency band and therefore provides a dimensionless descriptor of modal location and spectral width.

To evaluate the stability of the recovered modal profiles, we performed a wild-bootstrap refitting procedure around the fitted Tricomi model. A total of $B=500$ bootstrap refits were generated using Rademacher multipliers. Denoting by $\varepsilon'_{\mathrm{fit}}(\omega_k)$ and $\sigma_{\mathrm{fit}}(\omega_k)$ the fitted observables, bootstrap pseudo-data were generated in logarithmic coordinates by perturbing the fitted responses with sign-flipped residuals,
\begin{equation}
\begin{aligned}
\log_{10}\varepsilon'_{\ast}(\omega_k)
=
\log_{10}\varepsilon'_{\mathrm{fit}}(\omega_k)
+
v_k\Big[
\log_{10}\varepsilon'_{\mathrm{exp}}(\omega_k)
-
\log_{10}\varepsilon'_{\mathrm{fit}}(\omega_k)
\Big],
\end{aligned}
\label{eq:methods:bootstrap_eps}
\end{equation}
\begin{equation}
\begin{aligned}
\log_{10}\sigma_{\ast}(\omega_k)
=
\log_{10}\sigma_{\mathrm{fit}}(\omega_k)
+
v_k\Big[
\log_{10}\sigma_{\mathrm{exp}}(\omega_k)
-
\log_{10}\sigma_{\mathrm{fit}}(\omega_k)
\Big].
\end{aligned}
\label{eq:methods:bootstrap_sigma}
\end{equation}
where $v_k\in\{-1,+1\}$. The wild bootstrap is adopted in this context as it preserves the 
heteroscedastic structure of residuals across frequency decades, 
which is characteristic of logarithmic fitting in broadband 
impedance data. The pseudo-data were transformed back to linear scale and refitted with the same optimisation pipeline used for the original fit, using the original solution as warm start plus a small random perturbation within bounds.

For each bootstrap replicate, the corresponding in-band modal profiles were recomputed and normalised to unit area over $\log_{10}f$. Bootstrap envelopes were summarised pointwise by the median together with the 2.5th and 97.5th percentiles, and were used as empirical indicators of modal stability under data perturbation and refitting.

\subsection{Battery data and cycle-wise fitting in the complex domain}
\label{sec:methods:battery}

To complement the dielectric validation above, we finally considered electrochemical impedance data directly in the complex domain. The experimental spectra were taken from the dataset reported by Zhang et al.~\cite{Zhang2020}. For each selected cycle \(c\), the impedance data are denoted by
\begin{equation}
\mathcal{D}^{(c)}
=
\left\{
\left(f_k, Z_{\mathrm{exp}}^{(c)}(f_k)\right)
\right\}_{k=1}^{K},
\quad
Z_{\mathrm{exp}}^{(c)}(f_k)\in\mathbb{C}.
\end{equation}

The same bounded Tricomi construction introduced in Sect.~\ref{sec:model}, with \(F\) defined in Eq.~\eqref{eq:F_def}, was retained here, but embedded in a parallel-admittance representation suited to battery spectra. For the battery study, the model order was fixed to $N=3$, as this was the lowest order that consistently reproduced the observed multi-arc structure across cycles while preserving a stable and interpretable LF/MF/HF modal decomposition. Using the complementary factor
\begin{equation}
G_{a,b}(z)
:=
1-F_{a,b}(z)
=
\frac{1}{1+U(a,b,z)},
\label{eq:battery_G}
\end{equation}
and writing \(\omega_k = 2\pi f_k\), we define the internal admittance as
\begin{equation}
Y_{\mathrm{int}}^{(c)}(\omega_k;\theta^{(c)})
=
\frac{1}{R_0^{(c)}}
+
\sum_{n=1}^{N}
\frac{1}{R_n^{(c)}}
\,G_{a_n^{(c)},b_n^{(c)}}\!\left(\jmath\omega_k\tau_n^{(c)}\right),
\label{eq:battery_internal_admittance}
\end{equation}
and the model impedance as
\begin{equation}
Z_{\mathrm{mod}}^{(c)}(\omega_k;\theta^{(c)})
=
R_s^{(c)}
+
\jmath\omega_k L_s^{(c)}
+
\left(
Y_{\mathrm{int}}^{(c)}(\omega_k;\theta^{(c)})
\right)^{-1}.
\label{eq:battery_model}
\end{equation}

To preserve positivity constraints and improve numerical conditioning, the optimisation variables were expressed in mixed linear/logarithmic coordinates,
\begin{equation}
\begin{split}
\theta^{(c)}
=
\Big(
\log_{10}R_s,\,
\log_{10}L_s,\,
\log_{10}R_0,\,
\{a_n\}_{n=1}^{N},\,
\{b_n\}_{n=1}^{N},
\{\log_{10}\tau_n\}_{n=1}^{N},\,
\{\log_{10}R_n\}_{n=1}^{N}
\Big)^{(c)},
\end{split}
\label{eq:battery_theta}
\end{equation}
while the shape parameters were constrained to the same practical passive box used throughout the paper,
\begin{equation}
0.05 \le a_n \le 0.95,
\qquad
1.0002 \le b_n \le 1.95.
\label{eq:battery_shape_bounds}
\end{equation}
The bounds for \(R_s\), \(L_s\), \(R_0\), \(\tau_n\), and \(R_n\) were set from the measured impedance range and the acquisition bandwidth. To remove permutation ambiguity between dispersive branches, the fitted blocks were reordered after each evaluation according to increasing \(\tau_n\), with the associated parameters $a_n$, $b_n$ and $R_n$ permuted consistently.

Following the general fitting strategy of Sect.~\ref{sec:methods:fit}, the objective was defined through relative errors in the complex impedance ratio. For each frequency sample, we introduced
\begin{equation}
\rho_k^{(c)}
=
\frac{Z_{\mathrm{mod}}^{(c)}(\omega_k;\theta^{(c)})}
     {Z_{\mathrm{exp}}^{(c)}(\omega_k)},
\label{eq:battery_ratio}
\end{equation}
and the associated magnitude and phase residuals
\begin{equation}
e_{\mathrm{mag},k}^{(c)}
=
\log_{10}\left|\rho_k^{(c)}\right|,
\qquad
e_{\phi,k}^{(c)}
=
\arg\!\left(\rho_k^{(c)}\right).
\label{eq:battery_mag_phase_errors}
\end{equation}
This ratio-based formulation preserves relative sensitivity over the full spectrum and avoids over-emphasising only the largest-magnitude part of the impedance response.

Parameter estimation was carried out in two stages. First, a global differential-evolution search was performed on a log-subsampled subset of the spectrum in order to identify a robust initial parameter region. If \(\mathcal{I}_{\mathrm{DE}}\subset\{1,\dots,K\}\) denotes the subset of indices selected quasi-uniformly in \(\log_{10}f\), the pointwise loss was defined as
\begin{equation}
\ell_k^{(c)}
=
\left(e_{\mathrm{mag},k}^{(c)}\right)^2
+
\alpha_{\mathrm{DE}}
\left(e_{\phi,k}^{(c)}\right)^2,
\label{eq:battery_pointwise_loss}
\end{equation}
and the corresponding global objective as
\begin{equation}
\mathcal{L}_{\mathrm{DE}}^{(c)}
=
\left[
\frac{1}{|\mathcal{I}_{\mathrm{DE}}|}
\sum_{k\in\mathcal{I}_{\mathrm{DE}}}
\ell_k^{(c)}
\right]^{1/2},
\label{eq:battery_loss_de}
\end{equation}
with \(\alpha_{\mathrm{DE}}=0.5\). This stage was used only for coarse exploration of the non-convex parameter landscape.

Second, the differential-evolution solution was refined on the full frequency band by bounded nonlinear least squares with a robust soft-\(\ell_1\) loss. The residual vector supplied to the local solver was
\begin{equation}
\mathbf{r}_{\mathrm{LS}}^{(c)}
=
\begin{bmatrix}
e_{\mathrm{mag},1}^{(c)} \\
\vdots \\
e_{\mathrm{mag},K}^{(c)} \\
\sqrt{\alpha_{\mathrm{LS}}}\,e_{\phi,1}^{(c)} \\
\vdots \\
\sqrt{\alpha_{\mathrm{LS}}}\,e_{\phi,K}^{(c)}
\end{bmatrix},
\qquad
\alpha_{\mathrm{LS}}=1.
\label{eq:battery_residual_vector}
\end{equation}
The final fit quality therefore remained jointly controlled by relative magnitude and phase agreement over the full bandwidth. The post-fit scalar loss reported for each cycle was evaluated from the same log-ratio metric on the full grid.

For post-fit modal analysis, we use the complementary bounded mapping \(F_{a,b}=1-G_{a,b}\) as a geometric descriptor of the dispersive branch shape. This choice is not meant to represent the direct branch admittance contribution itself, which is modelled through \(G_{a,b}\) in Eq.~\eqref{eq:battery_internal_admittance}, but rather to provide a normalized in-band signature of modal location and width that can be compared consistently across cycles. Specifically, for the \(n\)-th branch at cycle \(c\), we define the raw profile
\begin{equation}
\phi_n^{(c)}(f)
=
\max\!\left\{
-\Im\!\left[
F_{a_n^{(c)},b_n^{(c)}}\!\left(\jmath 2\pi f\tau_n^{(c)}\right)
\right],
\,0
\right\},
\label{eq:battery_mode_raw}
\end{equation}
and its normalised version
\begin{equation}
\psi_n^{(c)}(f)
=
\frac{\phi_n^{(c)}(f)}
{\displaystyle
\int_{\log_{10}f_{\min}}^{\log_{10}f_{\max}}
\phi_n^{(c)}(\xi)\,d\log_{10}\xi }.
\label{eq:battery_mode_profile}
\end{equation}
By construction,
\begin{equation}
\int_{\log_{10}f_{\min}}^{\log_{10}f_{\max}}
\psi_n^{(c)}(f)\,d\log_{10}f = 1.
\label{eq:battery_mode_norm}
\end{equation}
Thus, \(\psi_n^{(c)}\) provides a dimensionless in-band descriptor of the spectral location and width of the \(n\)-th fitted contribution.

The characteristic frequency of each normalised profile was defined as
\begin{equation}
f_{p,n}^{(c)}
=
\operatorname*{arg\,max}_{f}\,
\psi_n^{(c)}(f).
\label{eq:battery_peak_frequency}
\end{equation}
This peak frequency was used as a scalar descriptor of modal location for cycle-to-cycle comparison.

\section{Results}
\subsection{Frequency-domain convergence of the Gauss--Stieltjes Foster approximation}
\label{sec:results:convergence}

In Fig.~\ref{fig:fig2}, we compare the anchored continuous response~\eqref{eq:H_two_plateau} against its passive $M$-term rationalisations obtained from the Gauss--Stieltjes construction of Sect.~\ref{sec:ss:ss_realisation}.
Results are reported for $M=5,15,45$ in two representative regimes, including magnitude and phase of $H(\jmath\omega)$, compact Nyquist insets, and the associated pole spectra shown as rug plots.

In the moderate-memory setting $a=0.7$, $b=1.7$, $\tau=1\,\mathrm{s}$ (left column), the lowest order $M=5$ already captures the global shape of both magnitude and phase, with residual discrepancies localised around the dispersive region.
Increasing to $M=15$ yields a close match across the displayed band, while $M=45$ produces a response that is visually indistinguishable from the continuous reference. This monotonic improvement is reflected by the embedded relative-error tables, which report both RMS and maximum errors for magnitude and phase.

\begin{figure*}[!h]
    \centering
    \includegraphics[width=\linewidth]{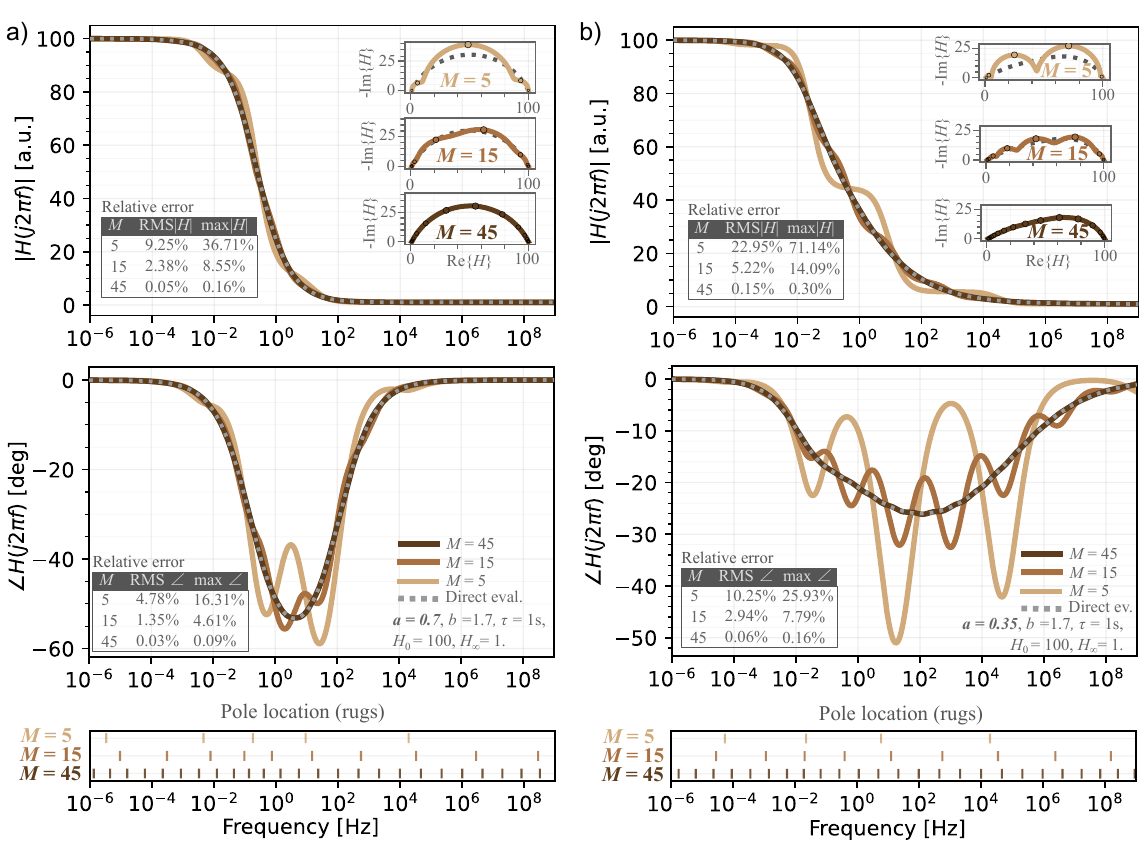}
    \caption{Frequency-domain convergence of the Gauss--Stieltjes Foster rationalisation embedded in the anchored two-plateau model~\eqref{eq:H_two_plateau}.
    Moderate-memory (left: $a=0.7$, $b=1.7$, $\tau=1\,\mathrm{s}$) and long-tail (right: $a=0.35$, $b=1.7$, $\tau=1\,\mathrm{s}$) regimes are shown.
    Direct evaluation of $F_{a,b}(z)$ is compared with passive $M$-term approximations ($M=5,15,45$; Sect.~\ref{sec:ss:ss_realisation}) in magnitude (top) and phase (middle), with Nyquist insets.
    Bottom panels show the corresponding discrete pole spectra (rug plots). Embedded tables report in-band relative error statistics (RMS/max) for magnitude and phase.}
    \label{fig:fig2}
\end{figure*}

The long-tail setting $a=0.35$, $b=1.7$, $\tau=1\,\mathrm{s}$ (right column) is substantially more demanding, as at low order $M=5$ the approximation exhibits clear deviations, most prominently in phase, and the mismatch extends over a broader portion of the transition. Increasing $M$ systematically reduces these discrepancies, with $M=45$ restoring excellent agreement throughout. The error tables confirm that the long-tail regime requires a higher order to reach the same accuracy level observed in the moderate-memory case, consistent with the broader log-rate content implied by longer memory.

The rug plots summarise the pole locations $\{\lambda_m\}$ induced by Gauss quadrature on the log-rate measure $p(u)\,\mathrm{d}u$ in Sect.~\ref{sec:ss:ss_realisation}, and should be interpreted together with the way the quadrature window is defined in practice. For numerical robustness, $p(u)$ is discretised over an extended log-rate interval that expands the displayed frequency band by $\pm 3$ decades on each side before applying the Lanczos/Stieltjes procedure. As a consequence, the method is explicitly allowed to place Gauss nodes (hence poles) outside the plotted band. In long-memory regimes this behaviour is expected: the spectral content that governs the in-band dispersion is spread over a wider range of log-rates, and a passive Stieltjes approximation must represent the out-of-band tails consistently with the in-band dynamics. This is why, for small $M$, some poles may fall outside the observable window (hence fewer visible rug marks), while increasing $M$ densifies the pole set and progressively resolves both the tail structure and the in-band dispersion, leading to the observed reduction of in-band errors.

Overall, Fig.~\ref{fig:fig2} demonstrates that the proposed Gauss--Stieltjes construction yields compact, passive and stable rational models whose accuracy improves systematically with $M$, while preserving the anchored two-plateau structure through $H_0$ and $H_\infty$.

\subsection{Multi-tissue validation in the complex domain}
\label{subsec:multi_tissue_complex_validation}

We considered heart muscle, breast fat, white matter and kidney as representative tissues spanning low-loss to highly conductive regimes and increasing dispersion complexity. Breast fat tests the model in a low-conductivity setting prone to overfitting artefacts, whereas heart muscle and kidney provide highly hydrated, strongly dispersive spectra where systematic bias becomes apparent. White matter introduces a structurally complex intermediate case, offering a stringent benchmark for mode identifiability. This selection enables a compact yet heterogeneous validation across distinct dielectric behaviours.

\begin{figure}[!h]
    \centering
    \includegraphics[width=\linewidth]{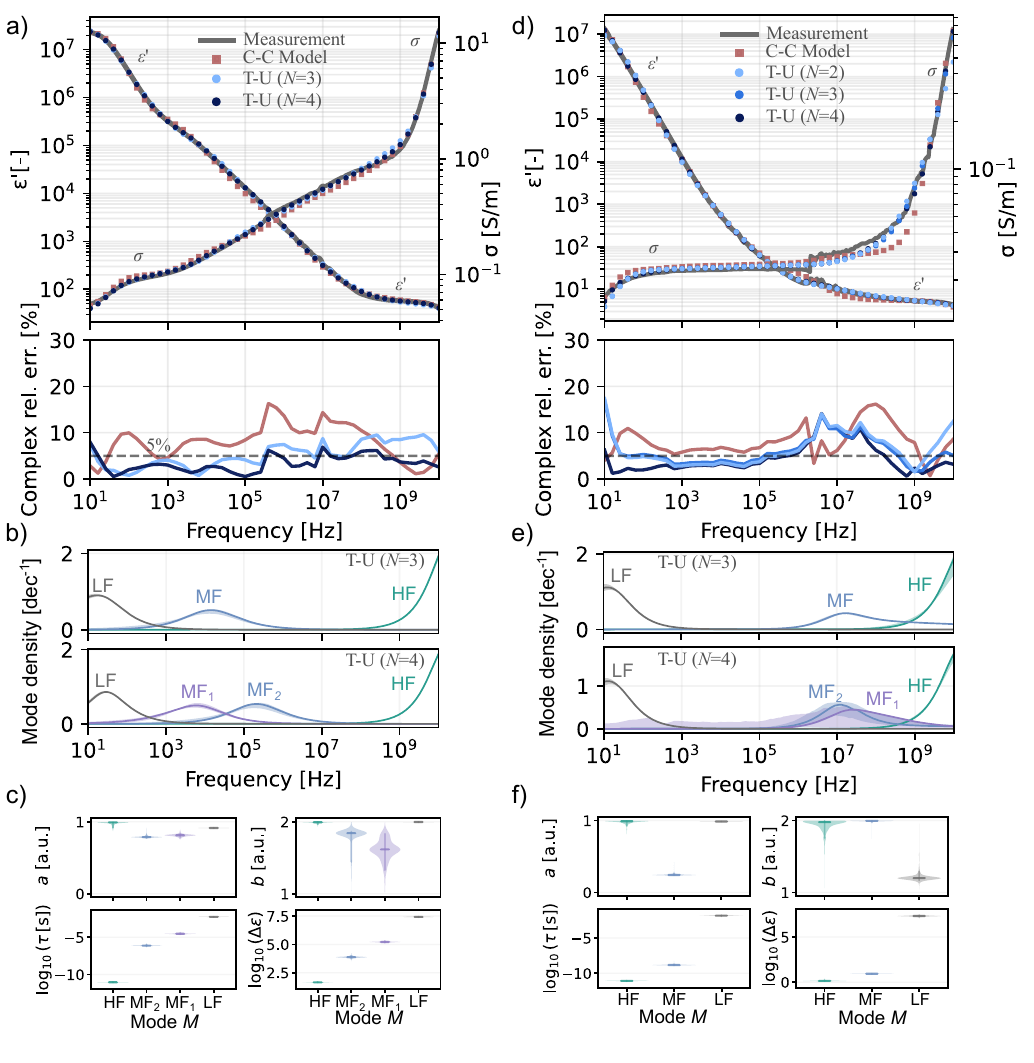}
    \caption{\textcolor{black}{Multi-tissue validation of the multi-block Tricomi--$U$ model. Panels (a)--(f) show model--measurement agreement across representative tissues (heart muscle, breast fat, white matter, kidney), comparing the fitted Tricomi--$U$ mixtures against the tabulated experimental Gabriel data and against the Cole--Cole-type baseline obtained from the IFAC-CNR online Gabriel resource. Overlaid mode-density curves highlight the frequency localisation and relative widths of the recovered relaxation modes. Increasing the number of blocks ($N=2,3,4$) refines mid-band structure where supported by the data, while bootstrap envelopes quantify mode stability across replicates.}}
    \label{fig:fig3}
\end{figure}

We quantified fitting performance directly in the complex domain using the relative error
\begin{equation}
    e(\omega)=\lvert \hat{\varepsilon}_{\mathrm{model}}(\omega)-\hat{\varepsilon}_{\mathrm{exp}}(\omega)\rvert/\lvert \hat{\varepsilon}_{\mathrm{exp}}(\omega)\rvert
\end{equation}
and reported the corresponding RMSE over a common evaluation grid. Across the four tissues, the proposed Tricomi--$U$ mixture achieved systematically lower complex-domain error than the Cole--Cole-type baseline obtained from the IFAC-CNR online Gabriel resource, while all errors were evaluated against the same tabulated experimental Gabriel measurements. In particular, using $N=4$ blocks reduced the RMSE from 8.303\% to 4.494\% in kidney, from 7.045\% to 2.575\% in white matter, from 8.782\% to 3.784\% in heart muscle, and from 8.679\% to 5.424\% in breast fat, corresponding to a relative reduction of approximately 38--63\% depending on tissue. Overall, the average RMSE was approximately halved (8.20\% to 4.07\%), highlighting the advantage of the proposed representation when evaluated against measured spectra in the complex domain.

Model order was selected using comparative information criteria computed from the same weighted logarithmic residual vector used for parameter estimation, as in~\eqref{eq:methods:residual_vector}. For kidney, white matter, and heart muscle, both AIC and BIC favoured $N=4$ over $N=3$, consistent with the pronounced RMSE drop when adding the fourth block (kidney: $\Delta\mathrm{BIC}=-40.84$; white matter: $\Delta\mathrm{BIC}=-26.47$; heart: $\Delta\mathrm{BIC}=-24.37$). In breast fat, the criteria exposed a low-loss regime where additional flexibility can be misleading: AIC preferred $N=4$ whereas BIC selected the more parsimonious $N=2$ model (BIC $-202.41$ for $N=2$ vs. $-199.20$ for $N=4$). In the remainder, we treat BIC as the primary guard against over-parameterization and complement it with a stability check based on bootstrap mode profiles.

More generally, the number of Tricomi blocks should not be increased
solely to reduce the residual, but should be constrained by the expected
number of physically relevant relaxation contributions within the measured
frequency band. In this sense, a fitted branch is interpreted as physically
meaningful only when it is supported by model-order criteria and displays
reproducible in-band modal placement under bootstrap perturbations.
To interpret the fitted response beyond a purely parametric description, we represent each Tricomi--$U$ block as a frequency-localised mode profile within the measured band, namely 10~Hz--10~GHz. For each block, we compute a loss-like profile over log-frequency, clipped to non-negative values, and normalise it to unit area in-band. This yields a mode-density curve that emphasises geometrical features of each mode, namely its characteristic location, where it concentrates in frequency, and its spectral width, enabling direct comparison across blocks, tissues and bootstrap replicates. When $N=4$ is supported (kidney, heart, white matter), the additional block typically refines mid-frequency content rather than introducing a qualitatively new LF or HF component, which aligns with the observed improvement in the transition region of the spectra.

Mode stability is assessed by repeating the modal profile extraction across bootstrap fits and reporting pointwise 95\% intervals (2.5--97.5 percentiles) together with the median profile. Across heart, kidney, and white matter, the dominant modes exhibit compact bootstrap bands and consistent placement, indicating that the modal decomposition is well constrained by the data even when raw parameters may be correlated. Conversely, in low-loss breast fat, increasing the number of blocks can yield an additional mode with broad bootstrap dispersion and inconsistent placement across replicates, mirroring the AIC/BIC disagreement and providing a clear signature of overfitting rather than a reproducible dispersive mechanism.

\subsection{Battery-ageing modelling in the complex domain}
\label{sec:results:battery_ageing_complex}

\begin{figure}[!htbp]
    \centering
    \includegraphics[width=0.5\linewidth,height=\textheight,keepaspectratio]{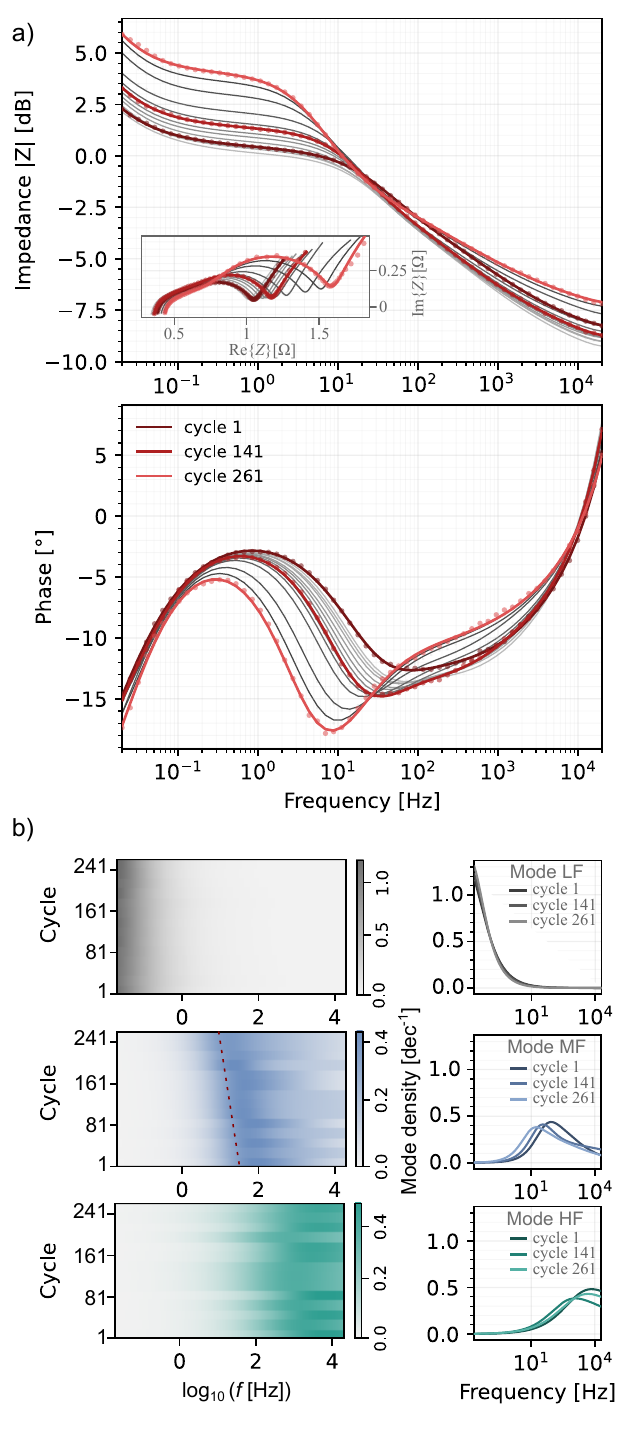}
    \caption{Battery-ageing modelling in the complex domain. Panel (a) compares measured and fitted spectra across selected ageing states in magnitude, phase and Nyquist representation. Panel (b) shows the corresponding normalised in-band mode-density profiles, resolving the LF, MF and HF contributions and their redistribution with cycling.}
    \label{fig:fig4}
\end{figure}

We next examine whether the same multi-block Tricomi framework can resolve ageing-induced spectral redistribution in electrochemical impedance data when the fitting is performed directly in the complex impedance domain.

Figure~\ref{fig:fig4} shows that the three-block Tricomi--$U$ representation accurately reproduces the measured impedance response over the full experimental bandwidth and across the selected ageing states. In panel (a), the fitted curves closely follow the experimental spectra in both impedance magnitude and phase, while the Nyquist inset confirms that the model also captures the progressive reshaping of the depressed arcs. As cycling proceeds, the low-frequency impedance level increases, the phase depression in the intermediate band becomes broader and deeper, and the Nyquist trajectories expand toward larger real impedances. Taken together, these features indicate that ageing is associated with a progressive redistribution of dissipative dynamics toward slower characteristic time scales, rather than with a mere uniform rescaling of the entire spectrum.

Panel (b) provides a modal interpretation of these spectral changes by resolving the fitted response into three normalized in-band mode-density profiles, labelled LF, MF, and HF. Importantly, these profiles are normalized and therefore encode spectral location and shape, not absolute modal amplitude. The LF mode remains concentrated at the lower edge of the observation window and decays monotonically across frequency, with only limited shape variation during cycling. This behaviour suggests that the slowest contribution remains only partially resolved within the measured bandwidth and acts primarily as a long-time-scale background component.

By contrast, the MF mode exhibits the clearest systematic evolution. Its ridge in the heat map progressively shifts toward lower frequencies, and the representative profiles for cycles 1, 141, and 261 become increasingly left-shifted and broader. This identifies the intermediate process as the main carrier of kinetic slowdown during ageing, consistent with a progressive increase in the effective relaxation time of interfacial and mesoscopic transport processes.

The HF mode is initially concentrated in the kilohertz range, but it broadens and moves toward lower frequencies at advanced cycling stages. As a result, the spectral separation between the HF and MF contributions becomes less pronounced, which explains the smoother and more distributed phase response observed in aged spectra. Overall,~\figurename~\ref{fig:fig4} supports an ageing scenario governed by modal redistribution across time scales: the slowest process remains marginally resolved, the intermediate mode undergoes the most systematic slowing down, and the fastest mode loses spectral compactness and partially merges with the mid-frequency dynamics.

\section{Discussion}

The analysis developed in this work establishes that the Tricomi-based construction is not merely an empirical fitting ansatz, but a bounded passive framework for anomalous relaxation with explicit spectral and realisation structure. Its main value lies in combining spectral flexibility with a passive and constructive mathematical structure. In contrast to purely phenomenological descriptions, the proposed bounded block is anchored by prescribed low-frequency and high-frequency limits, and admits a Stieltjes representation in the admissible parameter range. Therefore, it can be interpreted simultaneously as a dispersive response law, a passive spectral mixture, and a realizable finite-dimensional approximation. This combination is particularly relevant in applications where fitting accuracy alone is not sufficient. In these settings, passivity, interpretability and compatibility with circuit or state-space tools are equally important.

A useful way to position the present framework is relative to classical
fractional-order models. Fractional elements remain extremely effective at
reproducing broad power-law behaviour with very few parameters, and nothing
in the present results argues against their practical usefulness. However,
the Tricomi-based construction offers a different compromise. Since the
relation \(b=a+1\) recovers the Cole--Cole family exactly, the present
framework can be interpreted as a structurally grounded asymmetric extension
of Cole--Cole rather than merely as an unrelated non-fractional alternative.

{This positioning should not be understood as implying that established
generalized relaxation laws cannot describe asymmetric relaxation. For
example, the Havriliak--Negami model already provides a compact asymmetric
two-plateau response, with low- and high-frequency algebraic regimes governed
by different combinations of its fractional parameters. The distinction of
the present construction is instead one of parametrisation and realisation:
the low-frequency edge is controlled directly by \(b-1\), the
high-frequency edge directly by} \(a\) {and} \(\tau\) {sets the
frequency window in which the transition between the two plateaux occurs.}
By separating the exponents that govern the low-frequency and high-frequency
asymptotics, it allows asymmetric two-plateau transitions within a
non-fractional setting, while preserving a direct spectral interpretation and a constructive passive realisation.
{In contrast, Havriliak--Negami models are typically used as 
phenomenological transfer functions and do not provide a direct 
constructive route to passive finite-dimensional realizations. 
The present framework, through its Stieltjes structure, yields 
such realizations explicitly via Gauss--Stieltjes discretisation, 
with guaranteed positive poles and residues.}

{The separation of the exponents also gives the parameters an effective physical reading:}
\(b-1\) {describes the low-frequency, slow-process side of the response,}
\(a\) {describes the high-frequency, fast-process side, and}
\(\tau\) {sets the characteristic frequency window of the dispersive
transition. In dielectric and impedance spectroscopy, low-frequency features
are commonly associated with long-time transport, charge accumulation,
electrode or interfacial polarization, trapping, diffusion-limited dynamics,
or large-scale heterogeneity, whereas higher-frequency features are commonly
associated with local polarization, short-time relaxation, hopping-like
dynamics, or other small-scale mechanisms}~\cite{Kremer2003,Barsoukov2018,Lasia2014}.

{At the same time, especially in multi-block mixtures, this interpretation
should be understood at the level of effective spectral descriptors rather
than as a globally unique microscopic inversion of each individual parameter.
This is why the analysis emphasizes the fitted passive response and
bootstrap-stable in-band mode-density profiles, rather than raw parameter
values alone. Accordingly, these parameters should not be assigned to unique
microscopic mechanisms from impedance fitting alone; a case-specific physical
interpretation requires additional experimental evidence or dedicated future
studies, especially in heterogeneous systems where several relaxation
processes may overlap within the same frequency band.}

In this sense, the proposed model should not be viewed merely as a
replacement for fractional laws, but rather as an alternative modelling route
when one wishes to retain broad dispersive behaviour together with a more
explicit representation of relaxation structure. {Its practical value is
therefore not that it replaces Havriliak--Negami, Mittag--Leffler or
Prabhakar-type descriptions, but that it combines direct control of the two
asymptotic edges, bounded two-plateau anchoring, spectral positivity and a
Foster-type passive finite-dimensional realisation within the same
parametrisation.}

The rational approximation results reinforce this point. The Gauss--Stieltjes discretisation does not only approximate the continuous block accurately over frequency, but does so in a way that preserves positivity of poles and residues. This is important because it makes the passage from continuous special-function model to reduced passive implementation essentially seamless. The convergence patterns observed in moderate-memory and long-tail regimes also clarify an important practical point: the number of effective modes required for accurate approximation is not fixed a priori, but depends on the spectral breadth of the underlying kernel. In broader spectra, more modes are needed to resolve both the in-band dynamics and the out-of-band tails consistently.

{From the computational point of view, the framework has two distinct
levels of use. Direct evaluation of the bounded Tricomi block relies on the
confluent hypergeometric function} \(U(a,b,z)\){, and therefore has a
higher per-evaluation computational cost than classical Cole--Cole or
Havriliak--Negami laws, whose direct evaluation is essentially based on
complex powers and rational normalisations. The practical role of the
Gauss--Stieltjes construction is instead to provide a reduced passive
representation of the same response. After the offline construction of the
positive poles and residues, repeated frequency-domain evaluation or
time-domain implementation only requires a finite sum of first-order stable
terms. In this form, the retained order} \(M\) {provides a direct
accuracy--complexity trade-off, while preserving positivity, stability and
passivity by construction.}

The dielectric validations indicate that this additional flexibility is not merely formal. Across representative tissues, the multi-block Tricomi mixtures systematically improved the complex-domain fit with respect to the classical Cole--Cole baseline, especially in cases with pronounced multiscale dispersion. At the same time, the information-criterion and bootstrap analyses show that increased flexibility is not automatically beneficial. In low-loss settings such as breast fat, higher-order mixtures may introduce poorly constrained extra modes, whereas in more dispersive tissues the additional blocks are supported by both error reduction and modal stability. This suggests that the framework is capable of distinguishing between genuinely structured dispersion and over-parameterisation, which is an important practical advantage in broadband identification problems.

The battery analysis extends this interpretation to electrochemical impedance spectroscopy, where the same framework was applied directly in the complex impedance domain. Here the relevance of the modal representation is not that it provides absolute branch amplitudes, since the reported profiles are normalised in-band, but that it reveals how the spectral geometry of the fitted contributions evolves with cycling. The observed displacement of the dominant mid-frequency content toward lower frequencies, together with the broadening and partial overlap of the high-frequency contribution, is consistent with an ageing process in which dissipative dynamics progressively shift toward slower effective time scales. This type of description is difficult to obtain from a purely global fit and illustrates the interpretative value of the Tricomi decomposition beyond curve reproduction alone.

At the same time, the modal decomposition should be interpreted with appropriate caution. The present results do not claim that the recovered modes correspond one-to-one to uniquely identifiable microscopic physical processes, nor that they provide the true underlying relaxation spectrum of the material. Rather, they should be viewed as effective in-band dispersive descriptors induced by a passive model class. Their value lies in providing a compact and reproducible representation of spectral geometry, including modal location, width, and redistribution across conditions. In this form, they are useful both for comparative ageing analysis and for reduced-order modelling, where physically admissible and interpretable finite-dimensional surrogates are often more relevant than microscopic uniqueness.

Several limitations should nevertheless be kept in mind. First, although the passive parameter box and branch ordering substantially improve robustness, parameter identification in multi-block mixtures remains a non-convex problem, and distinct parameter sets may still generate very similar responses over finite bandwidths. The modal post-processing alleviates this issue at the level of spectral interpretation, but it does not remove identifiability limitations altogether. Second, the present validations were carried out with relatively simple additive branch architectures. More realistic electrochemical or biological systems may require stronger coupling between processes, additional constraints across datasets, or joint fitting across repeated conditions. Third, the modal profiles used in the battery study are intentionally normalised in-band; they are therefore suited to tracking spectral location and width, but not to quantifying absolute dissipative weight or to claiming unique microscopic mode attribution.

These limitations also point naturally toward future developments. One direction is to extend the present framework to richer multi-branch impedance architectures in which Tricomi blocks are combined with additional passive elements or constrained across state variables such as cycle number, temperature, or state of charge. Another is to investigate reduced-order implementations tailored to time-domain simulation, embedded solvers, or passive surrogate modelling. More broadly, the present results support the idea that special-function-based passive kernels can provide a useful middle ground between highly compact fractional phenomenology and fully unconstrained multi-pole fitting, retaining enough analytical structure to remain interpretable while remaining flexible enough to capture broad non-Debye dynamics in real data.

\section{Conclusions}

In this work we introduced a structured non-fractional framework for anomalous relaxation based on the Tricomi confluent hypergeometric function. Starting from the kernel representation of \(U(a,b,s\tau)\), we constructed a bounded two-plateau response through a M\"obius normalisation that enforces prescribed low-frequency and high-frequency limits while preserving a broad dispersive transition. Within the admissible parameter range, we established a Stieltjes representation with nonnegative spectral density, which in turn guarantees complete monotonicity, passivity, causality and compatibility with standard circuit and state-space descriptions. This provides a constructive alternative to fractional-order models when explicit spectral structure and passive realisability are required.

A second central result of the paper is that this passive spectral structure can be turned into a practical finite-dimensional representation. By exploiting the Stieltjes character of the bounded Tricomi block, we derived a Gauss--Stieltjes discretisation leading to Foster-type rational approximations and first-order state-space realisations with positive poles and residues. The numerical experiments showed that these approximations converge systematically with model order and remain robust even in long-tail regimes, thereby providing a direct bridge between the continuous spectral formulation and reduced passive implementations suitable for simulation and identification.

The empirical validations support the usefulness of the framework beyond its formal properties. In broadband dielectric data, multi-block Tricomi mixtures achieved systematically lower complex-domain error than the classical Cole--Cole baseline, while preserving interpretable modal structure and revealing when additional model order is genuinely supported by the data. In electrochemical impedance spectroscopy, the same framework reproduced the ageing-dependent evolution of battery spectra and resolved a clear redistribution of modal content toward slower characteristic time scales. These results indicate that the Tricomi-based construction is not only mathematically well posed, but also effective as an interpretable modelling tool for complex passive media.

Taken together, the proposed approach combines four features that are rarely available simultaneously within a single model class: bounded two-plateau behaviour, certified passivity, explicit spectral structure, and constructive rational realisation. This combination, in our view, is the main contribution of the present work. More broadly, the results suggest that special-function-based passive kernels can offer a viable alternative to fractional elements when one seeks both flexibility in spectral shaping and compatibility with classical circuit and system-theoretic tools.

Future work may extend this framework toward richer multi-branch electrochemical architectures, broader validation across impedance and dielectric datasets, and passive reduced-order implementations tailored to time-domain simulation and circuit solvers.
{In addition, it might be interesting to investigate qualitative and semi-quantitative 
links between the model parameters and underlying microscopic mechanisms, such as ion transport and interfacial polarisation.}
In these settings, the Tricomi construction provides a compact and physically grounded basis for modelling long-memory dynamics with broad and structured relaxation spectra.

\appendix

\section{Proof of Proposition~\ref{prop:rho_nonneg}}
\label{app:rho_proof}

We prove that, under the parameter conditions $a\in(0,1)$ and $b\in(1,2)$, the spectral density $\rho_F(x)$ defined in~\eqref{eq:rho_def} is nonnegative for all $x>0$.

We begin by evaluating the boundary value of $U(a,b,z)$ on the branch cut $z=-x+\jmath0^+$ for $x>0$. Using the connection formula~\eqref{eq:kummer-relation}, we write
\begin{align*}
U(a,b,-x+\jmath0^+)
= c_1\,M(a,b,-x)+c_2\,(-x+\jmath0^+)^{\,1-b}\,
  M(a-b+1,2-b,-x),
\end{align*}
where, for compactness of notation, we defined
\[
c_1=\frac{\Gamma(1-b)}{\Gamma(a-b+1)}, \qquad
c_2=\frac{\Gamma(b-1)}{\Gamma(a)}.
\]

On the principal branch, for $x>0$, the complex power expands as
\begin{equation}
(-x+\jmath0^+)^{\,1-b}
= x^{\,1-b} e^{\,\jmath\pi(1-b)} = x^{\,1-b} e^{\,\jmath\theta} \,,
\label{eqA:branch}
\end{equation}
where $\theta$ is the argument (phase) of the complex power on the principal branch. Substituting~\eqref{eqA:branch} into the previous expression yields the decomposition
\begin{equation}
U(a,b,-x+\jmath0^+) = A(x) + \jmath B(x).
\label{eqA:UAB}
\end{equation}
A short algebraic rearrangement gives
\begin{equation}
\begin{aligned}
A(x) = c_1\,M(a,b,-x) + c_2\,x^{1-b}\cos(\theta)
       M(a-b+1,2-b,-x),
       \label{eqA:A}
\end{aligned}
\end{equation}
\begin{equation}
\begin{aligned}
B(x) = c_2\,x^{1-b}\sin(\theta) M(a-b+1,2-b,-x).
       \label{eqA:B}
\end{aligned}
\end{equation}

We now determine the sign of $B(x)$. For $a\in(0,1)$ we have $\Gamma(a)>0$, and for $b\in(1,2)$ we also have $\Gamma(b-1)>0$, hence $c_2>0$. Moreover,
\[
\theta=\pi(1-b)\in(-\pi,0)
\quad\Longrightarrow\quad
\sin(\theta)<0.
\]
It remains to show that
\[
M(a-b+1,2-b,-x)>0,
\qquad x>0.
\]

Using Kummer's transformation~\cite{Slater1960confluent},
\begin{equation}
M(a-b+1,2-b,-x) = e^{-x}\,M(1-a,2-b,x),
\label{eqA:kummer}
\end{equation}
and noting that $1-a>0$ and $2-b>0$, the series expansion
\begin{equation}
M(1-a,2-b,x)
= \sum_{k=0}^{\infty}
  \frac{(1-a)_k}{(2-b)_k}\frac{x^k}{k!}
\label{eqA:series}
\end{equation}
has strictly positive coefficients for every $x>0$, so that
\[
M(a-b+1,2-b,-x)>0,\qquad x>0.
\]

At this point, having collected all sign contributions, let us consider $U=A+\jmath B$ from~\eqref{eqA:UAB} with $A,B$ given by~\eqref{eqA:A}--\eqref{eqA:B} and let us define $F_{a,b}=U/(1+U)$.
A direct computation yields
\begin{equation}
\Im\!\left\{\frac{U}{1+U}\right\}
= \frac{B}{(1+A)^2 + B^2}.
\label{eqA:ImF}
\end{equation}
Since the denominator of this latter is strictly positive and $B(x)<0$ for $x>0$, we obtain
\[
\Im\{F_{a,b}(-x+\jmath0^+)\} < 0,
\qquad x>0.
\]
By the definition~\eqref{eq:rho_def}, this implies $\rho_F(x) \ge 0$ for $x>0$, proving the nonnegativity of the spectral density.

By standard characterisations of Stieltjes functions and the Stieltjes--Perron inversion formula, the analyticity of $F_{a,b}$ on $\mathbb{C}\setminus(-\infty,0]$ together with the nonnegative boundary density $\rho_F$ yields the representation~\eqref{eq:F_stieltjes}. For $\Re\{z\}>0$, the integrand $\Re\{1/(x+z)\}>0$ ensures $\Re\{F_{a,b}(z)\}\ge0$, establishing that $F_{a,b}$ is positive-real.

\backmatter
\bmhead{Acknowledgments}
The work of I.~C. has been carried out in the framework of the activities of the Italian National Group of Mathematical Physics (GNFM), INdAM.

\bibliography{bib-fract}

\section*{Statements and Declarations}

\subsection*{Competing Interests}
The authors have no relevant financial or non-financial interests to disclose.

\subsection*{Availability of Data and Materials}
The datasets analysed during the current study are publicly available from the sources cited in the manuscript. The battery electrochemical impedance spectroscopy dataset associated with Zhang et al. is available in Zenodo (DOI: 10.5281/zenodo.3633835). The dielectric tissue data and reference Cole--Cole fits used in this work are available from the corresponding published sources cited in the manuscript. No new experimental data were generated in this study. Processed data supporting the findings of this study are available from the corresponding author upon reasonable request.

\end{document}